\documentclass[12pt]{article}
\usepackage{dsfont}
\usepackage{bm}
\usepackage{iftex}
\usepackage{stmaryrd}
\usepackage{amssymb}
\usepackage{comment}
\usepackage{mathtools}
\usepackage{stmaryrd}
\usepackage{bussproofs}
\usepackage{amsmath}
\usepackage{amsthm}
\usepackage{graphicx}
\usepackage{tikz}
\usepackage{a4wide}
\usepackage{nicefrac}

\newcommand{\weg}[1]{}

\newcommand{\ov}{\overline}
\newcommand{\half}{\nicefrac{1}{2}}
\newcommand{\eq}{\leftrightarrow}

\newcommand{\imp}{\rightarrow}
\newcommand{\Imp}{\Rightarrow}

\newcommand{\et}{\wedge}

\renewcommand{\phi}{\varphi}

\newcommand{\post}{\mathrm{post}}

\newcommand{\pre}{\mathrm{pre}}

\newtheorem{theorem}{Theorem}
\newtheorem{remark}[theorem]{Remark}
\newtheorem{definition}[theorem]{Definition}
\newtheorem{lemma}[theorem]{Lemma}
\newtheorem{example}[theorem]{Example}
\newtheorem{proposition}[theorem]{Proposition}
\newtheorem{corollary}[theorem]{Corollary}

\title{The Dynamic Turn in Paraconsistency}

\author{
Rafael Ongaratto\\
Institute of Philosophy and Human Sciences, UNICAMP\\
Campinas, Brazil\\
\texttt{ongarattorafa@gmail.com}
\and
Hans van Ditmarsch\\
University of Toulouse, IRIT-CNRS\\
Toulouse, France\\
Indian Institut of Technology (IIT)\\
Kanpur, India\\
\texttt{hansvanditmarsch@gmail.com}
}

\date{}

\begin{document}

\maketitle
\begin{abstract}
\noindent In this work we propose a dynamic turn in paraconsistency. We introduce AMLFI1, the action model extension of the paraconsistent logic LFI1. A special case is PALFI1, a paraconsistent logic of public announcements. It corresponds to another, recently published, paraconsistent public announcement logic: the differences in their axiomatizations are mutually admissible. We also introduce UMLFI1, that extends AMLFI1 with factual change. Soundness and completeness are proven for all logics, and all extend the epistemic paraconsistent logics KLFI1, KB4LFI1 and S5LFI1, known from the literature. With such dynamic epistemic paraconsistent logics we can formalize obtaining and resolving \textit{provisional} contradictions.
  \end{abstract}

\section{Introduction}\label{introduction}

What is the meaning of an eternal contradiction? Paraconsistent logic allows us to think about inconsistent but non-trivial models. However, offering a static image of paraconsistency will present, no matter how much we try to explain otherwise in the metalanguage, a static image of contradiction. The paradigmatic case of a static contradiction, looking through the lenses of dialetheists, would be the Liar Paradox \cite{sep-dialetheism}: a sentence that is contradictory regardless of time, and any further changes. Nevertheless, dialetheists' claim about real contradictions take it too far from what is necessary. A moderate position in paraconsistency is to adopt \textit{temporary} contradictions as paradigmatic. That is, inconsistencies appear at a certain point in the investigation and they can be later resolved--- the dialetheist would also agree that most contradictions are of this nature. 

\begin{quote}
    But the point we want to emphasize is that these
‘provisional contradictions’ are not dialetheias. In our view, they may be of
‘different kinds’ in the sense of having different causes. We list some of them:
(i) possible limitations of our cognitive apparatus; (ii) failure of measuring instruments and/or interactions of these instruments with phenomena; (iii) stages
in the development of theories; (iv) simply mistakes that in principle could be corrected later on. In all these cases, contradictions are related primarily to
knowledge and thought. This is what we call epistemic contradictions. \hfill \cite[p.\ 10-11]{carnielli2015logic}

\end{quote}

What contradictions of the type (i)-(iv) have in common, differently from the Liar Paradox, is that these contradictions presuppose \textit{change}. The model we should be looking for is one in which, at time $t_1$, the agent is uncertain about $p$, and after some action $k_1$, the agent may think he knows that $p$. Afterwards, at time $t_2$ and after some action $k_1$, the agent may have conflicting beliefs about $p$ and $\neg p$. For example, one of the applications of LFI1, a paraconsistent logic presented in this article, is to represent \textit{evolutionary} databases: databases that may change their integrity constraints over time, thus producing inconsistency \cite[p.\ 116]{Carnielli2000-CARFIA-2}. The logic LFI1 extends LP \cite{priest1979logic} with an inconsistency operator, resulting in a logic that is more expressive than LP and closer to classical propositional logic (denoted by CPC).

In general terms, a logic is called \textit{paraconsistent} if it has a unary operator of negation that is \textit{non-explosive}. That is, the principle \textit{ex falso sequitur quodlibet} fails in this logic. Paraconsistent logics started with Ja\'skowski and da Costa. In 1948, Ja\'skowski \cite{jaskowski1969propositional} proposed the first paraconsistent system that explicitly dealt with the problem of non-trivial and contradictory deductive theories, called \textit{discussive} logic --- also sometimes translated into english as \textit{discursive} logic. His system is called D2, a logic adequate to formalize discussions. Even though he did not give a complete axiomatization of his logic \cite[p.\ 1163]{omori2018axiomatizing}, his investigations were not paraconsistent by accident: he was explicitly dealing with the problem of separating the concept of contradiction from that of triviality. In 1964, da Costa \cite{da1974theory}, independently of Ja\'skowski's works, proposed a hierarchy of paraconsistent logics $C_n, 1 \leq n \leq \omega$.\footnote{Even though C$_\omega$ is not finitely trivializable, it is still a \textit{paraconsistent} logic, although it differs from the others C$_n$ systems because it is based on positive propositional intuitionistic logic \cite{carnielli1999limits}.}

Many years later, in the 2000s, Carnielli and Marcos \cite{carnielli2002taxonomy} generalized the idea behind $C_n$ systems, proposing what is known as Logics of Formal Inconsistency (LFIs). The introduction of consistency as a primitive operator $\circ$ is the foundational idea behind LFIs. Marcos was responsible for generalizing and developing the LFIs \cite{marcos2005logicsphd}, later on also developed by Carnielli and Coniglio \cite{carnielli2016paraconsistent}. The family of LFIs also incorporated paracompleteness in the Logics of Evidence and Truth (LETs) \cite{coniglio2024belnap, carnielli2015logic}, and from these works the epistemic approach to paraconsistency has been established. The combination of modal and paraconsistent logics has been extensively studied in the works of Bueno-Soler \cite{bueno2024many, bueno2012models, bueno2010two} and Coniglio \cite{coniglio2025combining,coniglio2024ivlev,carnielli2014swap}. A different approach to paraconsistency, in which negation is treated as a modal operator, has been extensively studied in the work of Marcos \cite{marcos2005nearly, marcos2005logics}.

Without incorporating dynamics within paraconsistent logics, the temporary aspect of contradictions that these logics intend to model under the epistemic approach of paraconsistency can only be explained in the metalanguage: at some time $t_1$, the state of the database is represented by the model $M_1$. After receiving new information, the database is represented by the model $M_2$. We cannot represent the update of model $M_1$ into model $M_2$ as the outcome of a process formalized in the logical language.
It is not possible to represent this model change using a static picture of paraconsistency. 

Such updates have been proposed in the Dutch school of dynamic logics, responsible for the so-called ``dynamic turn'' in epistemic logic \cite{baltag2012dynamic}. It brings exactly what paraconsistency still lacks: the representation of model change within logic. Let us succinctly describe its history. \textit{Dynamic epistemic logics} (DELs) were born out of the confluence of three areas of research during the 1980s and 1990s: dynamic logics, belief revision, and formal linguistics. \textit{Dynamic modal logics} \cite{harel:1984, pratt:1980, parikh:1978, goldblatt:1992} have modal formulas $[\pi]\beta$ read as `after every execution of program $\pi$, formula $\beta$ is true', where such programs $\pi$ are interpreted as transitions between states in a given relational Kripke model. \textit{Belief revision} \cite{agm:1985} by Alchourrón, Gärdenfors, and Makinson deals with expanding, contracting, or revising sets of beliefs described in first-order logic, according to certain set-theoretical operations. Thirdly, an area of formal linguistics that contributed to DEL's initial impulse was \textit{update semantics}, for example in the works of Groenendijk and Stokhof \cite{groenendijketal:1991} and of Veltman \cite{veltman:1996}. According to their analysis of natural language, when someone utters a sentence, this is a part of a process in which the sentence updates the information state. What is true in the next information state need not be a truth-functional consequence of what was true in the previous information state. In \cite{benthem1987}, van Benthem combines these three areas by suggesting to using dynamic modalities to model such information change. Further developments, sometimes independently, now follow up on each other fairly quickly. In public announcement logic (PAL) by Plaza \cite{plaza:1989}, $[\alpha]\beta$ means that after public announcement of $\alpha$, $\beta$ is true. Early developments of DEL also involve formalizing private announcements and factual change \cite{baltagetal:1998, gerbrandy:1999, linderetal:1995,hvd.thesis:2000,baltagetal:2004,hvdetal.hendricks:2003, renardel:1999,hvdetal.aamas:2005,jfaketal.lcc:2006}. Combinations of DEL with non-classical logics are more recent. In \cite{ma2014algebraic, balbiani2016intuitionistic}, the authors discuss an intuitionistic PAL. In \cite{bakhtiari2020bilattice, santos2020four}, the authors develop dynamic epistemic logics based on Belnap Dunn's 4-valued paraconsistent and paracomplete logic, which are also called paradefinite logics. In \cite{girard2016paraconsistent}, Girard and Tanaka present a public announcement version of LP, a three-valued paraconsistent logic. In the recent \cite{ongaratto2025dynamics}, Ongaratto presents a PAL extension of LFI1, a three-valued LFI.  


%
%
%

\paragraph*{Our contributions.}
We further extend epistemic (modal) extensions of the paraconsistent logic LFI1 with dynamic modalities for action models, and show that the resulting logics Action Model LFI1  (AMLFI1) and Update Model LFI1 (UMLFI1) are sound, strongly complete, compact, and decidable. The difference between AMLFI1 and UMLFI1 is that the latter involves action models with preconditions and postconditions, whereas the former involves action models with only preconditions, the better known standard in DEL. Completeness of these logics is shown by reduction rules, and a main technical novelty of our reduction axioms is that these are not equivalences but strong equivalences, based on congruences in this three-valued setting \cite{santos2020four,ongaratto2025}. We also show that the axiomatization of public announcement LFI1 (PALFI1) presented in this paper corresponds to the axiomatization of such a logic in \cite{ongaratto2025}, a system with different axioms and inference rules. We further note that few compactness and strong completeness results such as ours are known for many-valued DELs such as AMLFI1 and UMLFI1. Beyond technical contributions, the paper applies the dynamics of paraconsistency to the analysis of lying, being mistaken, and other byzantine behaviour, including recovery from such errors. The above-mentined factual change due to postconditions is essential to represent such recovery. We expect our results to be transferable to other dynamic logics of paraconsistency, and to have applications in distributed database management allowing local inconsistencies.

\weg{
\textcolor{red}{Our contributions in this paper are to propose Action Model and Update Model extensions of paraconsistent logics, in particular for LFI1. The technique of reduction by congruence relations, used in \cite{santos2020four} and in \cite{ongaratto2025} for the Public Announcement extension of LFI1, is extended and shown to hold equally well for Action Model and Update Model. Soundness and completeness are shown to hold for these cases. Furthermore, compactness and strong completeness are shown to hold, a relevant point that is yet underdeveloped in the literature for many-valued dynamic epistemic logics. Beyond technical contributions, the paper contributes to the discussion of the relation between lying and paraconsistency, and the possibilities of modelling that phenomenon. Finally, it provides a fresh perspective on paraconsistency by treating contradictions in dynamic terms.}

This technique is generalizable and available to be used in other contexts in which a congruence relation is definable.
}

\paragraph*{Overview of the paper.}
In Section~\ref{section2}, we present the logic LFI1 and its epistemic extensions that introduce the notion of knowledge in a paraconsistent context. In Section~\ref{section3}, we present the logic AMLFI1, a dynamic epistemic extension of LFI1 that allows us to represent refined model change. In Section~\ref{section4}, we show that the public announcement restriction PALFI1 of AMLFI1 is the same logic as the recent logic PALFI1$_C$ \cite{ongaratto2025dynamics}.
In Section~\ref{section5}, we consider UMLFI1, further extending AMLFI1 with factual change. Section~\ref{section.coup} applies the logics AMLFI1 and UMLFI1 to the representation of lying and other kinds of byzantine behaviour.

\section{The logic LFI1 and its modal extensions}\label{section2}
\subsection{Preliminaries}
Before introducing LFI1 and its modal extensions, it is important to introduce some general definitions that will be applied in all cases.

\begin{definition}[Language]\label{language}
    Let $P = \{p, q, r, \dots\}$ be a set of propositional variables, $N = \{a, b, c, \dots\}$ a finite set of agents, and $I$ a countably infinite index set representing an enumeration of actions (as will only be defined in in Section~\ref{section5}). The language $\mathcal{L}_{AM}$ is defined over the BNF form:
    \[\alpha::= p\ |\ \neg\alpha\ |\ \bullet\alpha\ |\ (\alpha \land \alpha)\ |\ K_i\alpha\ |\ [e]\alpha\] where $i \in A$ and $e \in I$.

    Elements of the language are \emph{formulas}, denoted by $\alpha, \beta, \gamma$. The language $\mathcal{L}_\bullet$ is the (propositional logical) $\neg\bullet\land$-fragment of $\mathcal{L}_{AM}$, and the language $\mathcal{L}_K$ is the $\neg\bullet\land K_i$-fragment of $\mathcal{L}_{AM}$. The \emph{signature} of a language is the set of its connectives. We consider the signatures $\Sigma_\bullet = \{\neg, \bullet, \land\}$, $\Sigma_K = \{\neg, \bullet, \land, \{K_i\}_{i \in N}\}$, and $\Sigma_{AM} = \{\neg, \bullet, \land, \{K_i\}_{i \in N}, \{e\}_{e \in I}\}$.
\end{definition}

     

\begin{definition}[Hilbert proof system]
A {\em (Hilbert) proof system} (axiomatization) $\mathcal{H}$ is a pair $\langle \mathcal{A}, \mathcal{R}\rangle$, where $\mathcal{A} \subseteq \mathcal{L}$ is the set of \emph{axioms} of \textbf{L}, and $\mathcal{R} \subseteq \mathcal{L}^*$ is the set of \emph{rules} of \textbf{L}.
\end{definition}

For a rule $\langle\alpha_1, \dots, \alpha_n, \alpha_{n+1}\rangle \in \mathcal{R}$ we write $\alpha_1, \dots, \alpha_n \Rightarrow \alpha_{n+1}$. Axioms $\alpha$ that are instantiations of a certain shape and rules $\alpha_1, \dots, \alpha_n \Imp \alpha_{n+1}$ that are instantiations of a certain shape are given names. An example axiom is $K_i \alpha \imp \alpha$ (T, truth axiom), and an example rule is $\alpha, \alpha\imp\beta \Imp \beta$ (MP, Modus Ponens).

\begin{definition}[Derivability]
Consider a Hilbert proof system $\mathcal{H} = \langle \mathcal{A}, \mathcal{R}\rangle$, and let $\Gamma \cup \{\alpha\} \subseteq \mathcal{L}$. The formula $\alpha$ is \emph{derivable} from $\Gamma$ in $\mathcal{H}$, notation $\Gamma \vdash_\mathcal{H} \alpha$, iff there is a sequence of formulas $\gamma_1, \dots \gamma_i$, such that $\gamma_i = \alpha$, and for all $k$ with $1 \leq k \leq i$ one of the following holds:
\begin{enumerate}
\item $\gamma_k \in \mathcal{A}$;
\item there exist $j_1, \dots, j_l < k$ such that $\langle \gamma_{j_1}, \dots, \gamma_{j_l}, \gamma_k\rangle \in \mathcal{R}$; 
\item $\gamma_k \in \Gamma$.
\end{enumerate}
\end{definition}
\noindent For $\emptyset \vdash_{\mathcal{H}} \alpha$ we write $\vdash_{\mathcal{H}} \alpha$.
\begin{definition}[Theorem]\label{derivability}
A formula $\alpha \in \mathcal{L}$ is a \emph{theorem} of proof system $\mathcal{H}$, iff $\vdash_{\mathcal{H}} \alpha$. 
    \end{definition}
Our proof theoretical terminology applies to the languages $\mathcal{L}_{AM}$, $ \mathcal{L}_K$, and $\mathcal{L}_\bullet$ of our contribution. However, there are some constraints and conventions in view of this usage. Derivability $\Gamma \vdash_\mathcal{H} \alpha$ in propositional logics allows the rule MP to be applied to hypotheses in $\Gamma$ and to axioms in $\mathcal A$ (and to formulas already obtained in that way). So its usage is not restricted to theorems. This is because MP is both validity-preserving and truth-preserving. Whereas all other derivation rules, in the more general modal logical setting, only apply to formulas that are theorems.

Note that we have presented the axioms and rules of our proof systems as {\em axiom schemata} (as above), formulated in terms of arbitrary formulas $\alpha,\beta,\gamma,\dots$, such that their instantiations are axioms and their instantiations apply the rules. We cannot formulate them with propositional variables $p,q,r,\dots$ instead. This is because dynamic epistemic logics do not usually have  the {\em substitution property} (if $\phi$ is a theorem and variable $p$ occurs in $\phi$, then uniform substitution of an arbitrary formula $\psi$ for $p$ in $\phi$ is also a theorem). We will see uniform substitution also fails for our DEL extensions of LFI1.

\begin{definition}[Logic]
A \textnormal{logic} \textbf{L} is a pair \textbf{L} $= \langle \mathcal{L}, \vdash\rangle$, in which $\vdash$ is a derivability relation. 
\end{definition}


\begin{definition}[Subsystem]\label{subsystem}
A logic $L_1$ is said to be a \textnormal{subsystem} (or a \textnormal{sublogic}) of $L_2$ when the following holds:\\
$(i)\ \mathcal{L}_{L_1} \subseteq \mathcal{L}_{L_2}$
\\$(ii)$ For all $\Gamma \cup \{\alpha\} \subseteq \mathcal{L}_{L_1}$,
\[\Gamma \vdash_{L_1} \alpha \text{ iff } \Gamma \vdash_{L_2} \alpha\]

\end{definition}

\begin{definition}[Valuation]\label{valuation}
Consider a set of truth-values $\mathrm{A} = \{1, \half, 0\}$ and a set $D = \{1, \half\}$ of designated truth-values. A \textnormal{valuation} $\vartheta$ is a function $\vartheta: P \mapsto \mathrm{A}$. In the propositional case, it is extended to arbitrary formulas in the following way: given an $n$-ary connective \#, $f_\#$ being its truth-function, and $\alpha_1, \dots, \alpha_n \in \mathcal{L}$, 
\[\vartheta(\#(\alpha_1, \dots, \alpha_n)) = f_\#(\vartheta(\alpha_1), \dots, \vartheta(\alpha_n))\] 

In the modal case, the valuation is also extended to formulas with modalities in a non-truth-functional way (see Definition~\ref{kripkesemantics}).

\end{definition}

\begin{definition}[Semantic Consequence]\label{semantics}
     Given $\Gamma \cup \{\alpha\} \subseteq \mathcal{L}$, $\alpha$ is a (semantic) consequence of $\Gamma$, denoted by $\Gamma \vDash \alpha$, iff for all valuations $\vartheta$, if $\vartheta[\Gamma] \subseteq \mathrm{D}$, then $\vartheta(\alpha) \in \mathrm{D}$. We say that a formula $\alpha$ is valid in L iff $\vDash \alpha$, that is, $\alpha$ is designated in every valuation $\vartheta$.

\end{definition}

    Given $\alpha \in \mathcal{L}$, the formula $\alpha$ is \textit{satisfiable} if there is some valuation $\vartheta$ such that $\vartheta(\alpha) \in \{1, \half\}$.
    
    \begin{definition}[Soundness]
Given a semantic consequence relation $\vDash$ and a syntactic consequence relation $\vdash$, $\vdash$ is sound iff the following holds for all $\Gamma \cup \{\alpha\} \subseteq \mathcal{L}$:

\[\Gamma \vdash \alpha \text{ implies } \Gamma \vDash \alpha\]
\end{definition}

\begin{definition}[Weak Completeness]
    Given a syntactic consequence relation $\vdash$, the relation $\vdash$ is weakly complete with respect to $\vDash$ iff the following holds for all $\alpha \in \mathcal{L}$:
    \[\vDash \alpha \text{ implies } \vdash \alpha\]
\end{definition}

\begin{definition}[Strong Completeness]
Given a syntactic consequence relation $\vdash$, the relation $\vdash$ is strongly complete with respect to $\vdash$ iff the following holds for all $\Gamma \cup \{\alpha\} \subseteq \mathcal{L}$:
\[\Gamma \vDash \alpha \text{ implies } \Gamma \vdash \alpha\]
\end{definition}

\begin{remark}
    Weak completeness does not always entail strong completeness. For example, LP is weakly complete with respect to valuations of classical logic, even though it is not strongly complete with respect to it. 
\end{remark}

    \begin{definition}[Compactness]
        A logic is compact with relation to its consequence relation $\vDash$ if the following holds for all $\Gamma \cup \{\alpha\} \in \mathcal{L}$:
        \[\Gamma \vDash \alpha \text{ implies that } \text{there exists } \Gamma' \text{ such that }\Gamma' \subseteq \Gamma, \Gamma' \text{ is finite and } \Gamma' \vDash \alpha\]
    \end{definition}

\begin{definition}[Decidability]
    A logic is decidable with relation to a consequence relation $\vDash$ iff there is an effective procedure that determines in a finite number of steps, given $\alpha \in \mathcal{L}$, whether $\alpha$ is satisfiable or not. That is, if $\alpha$ is satisfiable, it is possible to determine in a finite number of steps a valuation $\vartheta$ such that $\vartheta(\alpha) \in \{1, \half\}$. If $\alpha$ is not satisfiable, it is possible to determine in a finite number of steps a countermodel of $\alpha$, that is, a valuation $\vartheta$ such that $\vartheta(\alpha) = 0$.
\end{definition}

For the logics mentioned in this paper, consider the following abbreviations: CPL stands for ``Classical Propositional Logic'', and LP stands for ``Logic of Paradox''.

\subsection{Paraconsistent logic LFI1}

In this section we present the basics of the paraconsistent propositional logic LFI1. The logic LFI1 is a well-known three-valued logic and offers an intuitive interpretation of paraconsistency. This logic was introduced in a different signature as J3 \cite{d1985completeness}, and it has interesting relations to well-known logics such as LP, which is a proper fragment of LFI1 \cite[p.\ 151]{carnielli2016paraconsistent}. It is  functionally equivalent to \L $_3$ \cite[p.\ 142]{carnielli2016paraconsistent}. 

     

The language of LFI1 is $\mathcal{L}_\bullet$ (Definition~\ref{language}), where $\bullet\alpha$ means ``$\alpha$ is inconsistent''. Below, some definable connectives are presented, together with their informal names. 

 \begin{definition}\label{definability}
     The connectives $\circ, \bot, \sim, \lor, \rightarrow, \leftrightarrow, \equiv$ are definable through the following abbreviations:
\[
\begin{array}{llll}
  \text{Consistency}  &    \circ\alpha & := & \neg\bullet\alpha\\
\text{Disjunction}     &   \alpha \lor \beta & := & \neg(\neg\alpha \land \neg\beta)\\
 \text{Falsum} &       \bot & := & \alpha \land \neg\alpha \land \circ\alpha         \\ \text{Strong negation} & \sim\alpha & :=  & \neg\alpha \land \circ\alpha           \\ \text{Implication} &
    \alpha \rightarrow \beta & := & \sim\alpha \lor \beta \\ \text{Equivalence}
  &  \alpha \leftrightarrow \beta & := & (\alpha \rightarrow \beta) \land (\beta \rightarrow \alpha)\\ \text{Strong equivalence}
  &  \alpha \equiv \beta & := & (\alpha \leftrightarrow \beta) \land (\neg\alpha \leftrightarrow \neg\beta)
     
\end{array}
\]
 \end{definition}

The logic LFI1 is a Logic of Formal Inconsistency (LFI), which means that it includes a unary operator that \textit{recovers} explosion in a controlled form. A logic is an LFI if it is paraconsistent with relation to $\neg$ and it includes a unary recovery\footnote{The unary operator is called a ``recovery'' operator because it \textit{recovers} classical reasoning by restrictively recovering the Principle of Explosion (PE) in paraconsistent logics (see \cite{carnielli2020recovery}).} operator $\circ$ such that \cite[p.\ 32]{carnielli2016paraconsistent}:
\[\begin{array}{lll}
   (i)  &  \alpha, \neg\alpha, \circ\alpha \vdash \beta, & \text{for all formulas } \alpha, \beta\\
    (ii) & \alpha, \circ\alpha \nvdash \beta, & \text{for some formulas } \alpha,\beta 
    \\(iii)& \neg\alpha, \circ\alpha \nvdash \beta, & \text{for some formulas } \alpha, \beta
\end{array}\]

We now present the axiomatization of LFI1 of \cite[p.\ 162]{carnielli2016paraconsistent}.

\begin{definition}
      The logic LFI1 is defined over the language $\mathcal{L}_\bullet$ by the Hilbert Calculus \cite[p.\ 160]{carnielli2016paraconsistent}: \[
\begin{array}{l@{\qquad}l}
\begin{array}{l l}
A1  & \alpha \rightarrow (\beta \rightarrow \alpha)\\
A2  & (\alpha \rightarrow (\beta \rightarrow \gamma))
      \rightarrow ((\alpha \rightarrow \beta) \rightarrow (\alpha \rightarrow \gamma))\\
A3  & \alpha \rightarrow (\beta \rightarrow (\alpha \land \beta))\\
A4  & (\alpha \land \beta) \rightarrow \alpha\\
A5  & (\alpha \land \beta) \rightarrow \beta\\
A6  & \alpha \rightarrow (\alpha \lor \beta)\\
A7  & \beta \rightarrow (\alpha \lor \beta)\\
A8  & (\alpha \rightarrow \gamma)
      \rightarrow ((\beta \rightarrow \gamma)
      \rightarrow ((\alpha \lor \beta) \rightarrow \gamma))
\end{array}
&
\begin{array}{l l}
A9  & \alpha \lor (\alpha \rightarrow \beta)\\
A10 & \alpha \lor \neg\alpha\\
A11 & \circ\alpha \rightarrow (\alpha \rightarrow (\neg\alpha \rightarrow \beta))\\
A12 & \bullet\alpha \rightarrow (\alpha \land \neg\alpha)\\
A13 & \neg\neg\alpha \leftrightarrow \alpha\\
A14 & \neg(\alpha \lor \beta) \leftrightarrow (\neg\alpha \land \neg\beta)\\
A15 & \neg(\alpha \land \beta) \leftrightarrow (\neg\alpha \lor \neg\beta)\\
A16 & \neg(\alpha \rightarrow \beta) \leftrightarrow (\alpha \land \neg\beta)\\
MP  & \alpha,\, \alpha \rightarrow \beta \Rightarrow \beta
\end{array}
\end{array}
\]
\end{definition}

The notion of an LFI1-derivation $\vdash_{LFI1}$ is according to Definition \ref{derivability}.

\begin{remark}
    Despite containing only two axioms concerning the inconsistency operator, namely, axioms A11 and A12, LFI1 is strong enough to be characterizable as a three-valued logic that extends positive classical logic (that is, the negation-less fragment of classical logic) with a paraconsistent negation and an inconsistency operator. Axiom A11 says that if $\alpha$ is consistent, then $\alpha$ and $\neg\alpha$ cannot both be the case. Axiom A12 says that if $\alpha$ is inconsistent, then $\alpha$ is contradictory. Together, these two axioms ensure that $\alpha$ is consistent iff $\alpha$ is not contradictory. 
\end{remark}


\begin{proposition}[Replacement of Equivalents \normalfont{\cite[p.\ 181]{carnielli2016paraconsistent}}]\label{replacement}
    The logic LFI1 satisfies Replacement of Equivalents (RE) with respect to $\equiv$. For all $\theta, \gamma, \alpha \in \mathcal{L}$,
    \[\begin{array}{ll}
   RE & \theta \equiv \gamma \Rightarrow  \alpha \equiv \alpha(\theta/\gamma) 
    \end{array}\]
\end{proposition}

The strong equivalence $\equiv$, that is definable (Definition~\ref{definability}) is important in our paper, because it is characterized as a \textit{congruence relation}. Not all many-valued logics have such congruence relations, for example the logics K$_3$ and LP do not have it. Congruence relations are an essential feature for the algebraizability of  logics.

\begin{proposition}
    The connective $\equiv$ is a congruence relation on LFI1.
\end{proposition}

\begin{proof} 
Strong equivalence $\equiv$ satisfies the following conditions in LFI1 \cite[p. 130]{carnielli2016paraconsistent}, where $\#$ is an arbitrary $n$-ary connective in the signature:
    \[
    \begin{array}{lll}
       (i)  & \alpha \equiv \alpha& \text{Reflexivity}\\
       (ii)  & \alpha \equiv \beta \Rightarrow \beta \equiv \alpha&\text{Symmetry}\\
       (iii) & \alpha \equiv \beta, \beta \equiv \gamma \Rightarrow\ \alpha \equiv \gamma & \text{Transitivity}\\
       (iv) &  \alpha_1 \equiv \beta_1, ...,\ \alpha_n \equiv \beta_n \Rightarrow \#(\alpha_1, ..., \alpha_n) \equiv \#(\beta_1, ..., \beta_n) & \text{Congruenciality}
    \end{array}
\]
As in LFI1 the connectives $\#$ are $\neg,\et,\bullet$ we need to observe that the following instantiations of $(iv)$ are derivable.
\[
\begin{array}{ll}
     \neg R & \alpha \equiv \beta \Rightarrow \neg\alpha \equiv \neg\beta  \\
     \bullet R&  \alpha \equiv \beta \Rightarrow \bullet\alpha \equiv \bullet\beta\\
     \land R &  \alpha_1 \equiv \beta_1, \alpha_2 \equiv \beta_2 \Rightarrow (\alpha_1 \land \alpha_2) \equiv (\beta_1 \land \beta_2)
\end{array}\]
This is easy to verify and thus ends the proof.
\end{proof}

\begin{definition}
    Let $\mathcal{M} = (\mathcal A, \mathrm{D})$ be a matrix such that $A$ and $D$ are defined as in Definition~\ref{valuation}. The truth tables of $\mathcal{M}$ are \cite[p.\ 159]{carnielli2016paraconsistent}:
\begin{table}[h]
\begin{tabular}{|l|l|l|l|l|l|l|l|}
\cline{1-4} \cline{6-8}
$\land$ & 1   & $\half$ & 0 &  &     & $\neg$ & $\bullet$ \\ \cline{1-4} \cline{6-8} 
1       & 1   & $\half$ & 0 &  & 1   & 0      & 0         \\ \cline{1-4} \cline{6-8} 
$\half$     & $\half$ & $\half$ & 0 &  & $\half$ & $\half$    & 1         \\ \cline{1-4} \cline{6-8} 
0       & 0   & 0   & 0 &  & 0   & 1      & 0         \\ \cline{1-4} \cline{6-8} 
\end{tabular}
\end{table}

Some of the definable connectives are characterized by the truth functions below.

\begin{table}[h]
\begin{tabular}{|l|l|l|l|l|l|l|l|l|l|l|l|l|l|l|l|l|l|l|}
\cline{1-4} \cline{6-9} \cline{11-14} \cline{16-19}
$\lor$ & 1 & $\half$ & 0   &  & $\rightarrow$ & 1 & $\half$ & 0 &  & $\equiv$ & 1 & $\half$ & 0 &  &     & $\circ$ & $\sim$ & $\bot$ \\ \cline{1-4} \cline{6-9} \cline{11-14} \cline{16-19} 
1      & 1 & 1   & 1   &  & 1             & 1 & $\half$ & 0 &  & 1        & 1 & 0   & 0 &  & 1   & 1       & 0      & 0      \\ \cline{1-4} \cline{6-9} \cline{11-14} \cline{16-19} 
$\half$    & 1 & $\half$ & $\half$ &  & $\half$           & 1 & $\half$ & 0 &  & $\half$      & 0 & $\half$ & 0 &  & $\half$ & 0       & 0      & 0      \\ \cline{1-4} \cline{6-9} \cline{11-14} \cline{16-19} 
0      & 1 & $\half$ & 0   &  & 0             & 1 & 1   & 1 &  & 0        & 0 & 0   & 1 &  & 0   & 1       & 1      & 0      \\ \cline{1-4} \cline{6-9} \cline{11-14} \cline{16-19} 
\end{tabular}
\end{table}

For all $c \in \Sigma_\bullet$, let $f_c$ be the truth function it denotes in $\mathcal{M}$. A function $\vartheta: \mathcal{L}_\bullet \imp A$ is an LFI1-\emph{valuation} if it satisfies 
for all $c \in \Sigma$ that:

\[\begin{array}{c}
           \vartheta(c(\alpha_1, \alpha_2, ..., \alpha_n)) = f_c(\vartheta(\alpha_1), ..., \vartheta(\alpha_n))
\end{array}    \]  

The notion of logical consequence, denoted by $\vDash_{LFI1}$, is according to Definition~\ref{semantics}.
\end{definition}

\begin{theorem}[Soundness and Strong Completeness \normalfont{\cite{Carnielli2000-CARFIA-2})}]\label{completenesslfi1}
Consider $\Gamma \cup \{\alpha\} \subseteq \mathcal{L}_\bullet$. Then, $$\Gamma \vdash_{LFI1} \alpha \text{ iff } \Gamma \vDash_{LFI1} \alpha$$
\end{theorem}

In LFI1, the usual logical equivalence does not induce  a congruence. Consider $\bullet\alpha \leftrightarrow (\alpha \land \neg\alpha)$. Indeed, from axioms 11 and 12 it follows that $\vdash \bullet\alpha \leftrightarrow (\alpha \land \neg\alpha)$. However, using Soundness from Theorem \ref{completenesslfi1} it is straightforward to prove that $\nvdash \neg\bullet\alpha \leftrightarrow \neg(\alpha \land \neg\alpha)$: consider valuation $\vartheta(\alpha) = \half$. Then $\vartheta(\neg\bullet\alpha) = 0$, whereas $\vartheta(\neg(\alpha \land \neg\alpha)) = \half$. We therefore need the stronger notion of syntactical equivalence $\equiv$ that we defined as {\em strong equivalence}. Semantically we can prove that $\equiv$ is a strong equivalence by showing that $\vartheta(\alpha \equiv\beta) \in D$ iff $\vartheta(\alpha)=\vartheta(\beta)$. For a detailed exposition on such strong or {\em strict equivalences} (\textit{congruence relations}), in the context of algebraic logics, see also \cite{blok1989algebraizable}.

The logic LFI1 is also decidable and compact, which has been proved in general for any \textit{finitely}-many-valued logic in \cite[p.\ 34, p.\ 172]{gottwald2001treatise}.


\begin{lemma}\label{equivalences}
    The following strict equivalences hold in LFI1:
    \[
\begin{array}{l@{\qquad}l}
\begin{array}{l l}
1. & \vdash \alpha \land (\alpha \rightarrow \beta) \equiv \alpha \land \beta\\
2. & \vdash (\alpha \land \beta) \rightarrow \gamma \equiv \alpha \rightarrow (\beta \rightarrow \gamma)
\end{array}
&
\begin{array}{l l}
3. & \vdash \alpha \rightarrow (\beta \rightarrow \gamma) \equiv \beta \rightarrow (\alpha \rightarrow \gamma)\\
4. & \vdash \alpha \rightarrow (\beta \lor \gamma) \equiv (\alpha \rightarrow \beta) \lor (\alpha \rightarrow \gamma)
\end{array}
\end{array}
\]
\end{lemma}
\begin{proof} 
Using completeness of LFI1 from Theorem \ref{completenesslfi1}, it is straightforward to check the congruences in each case via truth-tables.
\end{proof}

\begin{remark}
Despite being simple cases of equivalences in classical logic, finding congruences in LFI1 becomes a non-trivial task: not every equivalence is a congruence. Consider, for example, the case of the equivalence $\bullet\alpha \leftrightarrow (\alpha \land \neg\alpha)$. Indeed, from axioms 11 and 12, it is the case that $\vdash \bullet\alpha \leftrightarrow (\alpha \land \neg\alpha)$ (left-to-right is just an application of A12; right-to-left is obtained \textit{ad absurdum} supposing that $\circ\alpha$ is the case and an application of A11).
However, using soundness of LFI1 it is straightforward to prove that $\nvdash \neg\bullet\alpha \leftrightarrow \neg(\alpha \land \neg\alpha)$: consider a valuation $\vartheta$ such that $\vartheta(\alpha) = \half$. Then, $\vartheta(\neg\bullet\alpha) = 0$, whereas $\vartheta(\neg(\alpha \land \neg\alpha)) = \half$.
\end{remark}

\begin{remark}
     The Logic of Paradox (LP) is the $\neg\land$-fragment of LFI1. Extending LP with an inconsistency operator $\bullet$ results in LFI1 \cite{ongaratto2025}. The logic LFI1 can define many connectives that are not definable in LP, as we can see from Definition \ref{definability}. This makes LP and LFI1 two very distinct logics. The logic LP has the same validities as classical propositional logic CPL: $\vDash_{LP}\alpha$ iff $\vDash_{CPL} \alpha$) \cite[p.\ 1296]{golan2023simple}, but it still has a different consequence relation. In contrast, LFI1 and LP do not have the same validities. Although $\vDash_{LFI1} \alpha$ implies $\vDash_{CPL} \alpha$, because LFI1 is a subclassical logic, the converse does not hold. A straightforward counterexample is the CPL tautology $\alpha \rightarrow (\neg\alpha \rightarrow \beta)$ which is invalid in LFI1. Furthermore, implications in LP and LFI1 behave differently. Whereas the conditional (the implication) in LP is definable as $\alpha \rightarrow_{LP} \beta :=\neg\alpha \lor \beta$, resulting in a non-detachable conditional, LFI1 can define stronger implications such as $\rightarrow$ that are detachable. 
 \end{remark}

\subsection{Modal extensions of LFI1}

We now continue with the epistemic modal extensions of LFI1. They are defined over the $\mathcal{L}_K$ language of Definition~\ref{language}. The \textit{dual} of $K_i$, $\hat{K}_i$, is definable as $\hat{K}_i \alpha := \neg K_i \neg\alpha$ \cite{bueno2025towards}. 
\begin{definition}\label{klfi1}
      The logic KLFI1 extends LFI1 with the following axioms and rules, for all $i \in N$: 
\[\begin{array}{l l}
 \mathrm{K} &  K_i(\alpha \rightarrow \beta) \rightarrow (K_i\alpha \rightarrow K_i\beta)    \\

\mathrm{K}_1 & K_i(\alpha \rightarrow \beta) \rightarrow (\hat{K}_i\alpha \rightarrow \hat{K}_i\beta)\\

 \mathrm{K}_2 & \hat{K}_i(\alpha \lor \beta) \rightarrow (\hat{K}_i\alpha \lor \hat{K}_i\beta)
 \\
 \mathrm{K}_3 & (\hat{K}_i\alpha \rightarrow K_i\beta) \rightarrow K_i(\alpha \rightarrow \beta)\\
  \mathrm{Nec} & \alpha \Rightarrow K_i\alpha 
\end{array}\]
\end{definition}

Because contraposition does not hold, KLFI1 requires additional axioms K$_1$-K$_3$, apart from K, to ensure completeness \cite[p. 3]{bueno2012models}.

  \begin{definition}\label{kripkesemantics}
        A Kripke model for KLFI1 defined over $\mathcal{L}_K$ is a structure $\langle W, \{\sim_i\}_{i \in N}, \vartheta\rangle$, where each $\sim_i$ is a relation on $W$ for all $i$ $\in$ N, and where $\vartheta$ is a set of valuations $\vartheta_w$ for each $w \in W$ such that $\vartheta_w$ extends an LFI1-valuation with condition:
\[\begin{array}{c}
                  \vartheta_{w}(K_i\alpha) = inf\{\vartheta_{w'}(\alpha): w \sim_i w'\}     
\end{array}    \]    
 A formula $\alpha$ is \textit{valid} on a model $\mathcal{M}$, denoted by $\mathcal{M} \vDash \alpha$, if $\vartheta_w(\alpha) \in \{1, \half\}$ for all $w \in W.$ An S5LFI1-valuation is an LFI1-valuation that satisfies the additional clause above for modal operators. 
 
            \end{definition}
Soundness and (strong) completeness of KLFI1 with respect to its Kripke semantics are proven in \cite[p.\ 227]{bueno2024many}, as well as that of its normal modal extensions. The strict equivalence defined for LFI1 preserves its properties in KLFI1 and its epistemic extensions \cite{bueno2025towards}. 

\begin{definition}
   We can extend KLFI1, the minimal \textit{normal modal logic} \cite{bueno2025towards} based on LFI1, with axioms, for $i \in N$:
\[
\begin{array}{l@{\qquad}l}
\begin{array}{ll}
\mathrm{T} & K_i\alpha \rightarrow \alpha\\
\mathrm{D} & K_i\alpha \rightarrow \hat{K}_i\alpha\\
\mathrm{B} & \alpha \rightarrow K_i\hat{K}_i\alpha
\end{array}
&
\begin{array}{ll}
\mathrm{4} & K_i\alpha \rightarrow K_iK_i\alpha\\
\mathrm{5} & \hat{K}_i\alpha \rightarrow K_i\hat{K}_i\alpha
\end{array}
\end{array}
\]
\noindent
As above for KLFI1, we need to extend each system with the following dual axioms \cite[p. 26]{bueno2025towards}. Note that the contraposition D$'$ is D again. It is therefore lacking below.

\[
\begin{array}{l@{\qquad}l}
\begin{array}{ll}
\mathrm{T'} & \alpha \rightarrow \hat{K}_i\alpha\\
\mathrm{B'} & \hat{K}_iK_i\alpha \rightarrow \alpha
\end{array}
&
\begin{array}{ll}
\mathrm{4'} & \hat{K}_i\hat{K}_i\alpha \rightarrow \hat{K}_i\alpha\\
\mathrm{5'} & \hat{K}_iK_i\alpha \rightarrow K_i\alpha
\end{array}
\end{array}
\]
\end{definition}

\begin{proposition}[Deduction Theorem \normalfont{\cite[p.\ 11]{bueno2025towards})}] Consider $\Gamma \cup \{\alpha, \beta\}\subseteq \mathcal{L}_K$.
Then, 
\[\Gamma \cup \{\alpha\} \vdash_{KLFI1} \beta \text{ implies } \Gamma \vdash_{KLFI1} \alpha \rightarrow \beta\]
\end{proposition}

The same holds for other normal modal extensions of KLFI1.

\begin{corollary}[Compactness]\label{compactness} Consider $\Gamma \subseteq \mathcal{L}_K, \alpha \in \mathcal{L}_K$, and $\Gamma'$ such that $\Gamma' \subseteq \Gamma$ and $\Gamma'$ is finite. 

\[\Gamma \vDash_{KLFI1} \alpha \text{ implies } \Gamma' \vDash_{KLFI1} \alpha\]
    
\end{corollary}

\begin{proof}
    Compactness follows from strong completeness. Suppose that $\Gamma \vDash_{KLFI1}\alpha$. By strong completeness, $\Gamma \vdash_{KLFI1}\alpha$. By definition of KLFI1-derivation, there is a subset $\Gamma' \subseteq \Gamma$ such that $\Gamma'$ is finite and $\Gamma' \vdash_{KLFI1} \alpha$. By soundness, this implies that $\Gamma' \vDash \alpha$. 
\end{proof}

\begin{proposition}[Decidability {\normalfont\cite[p.~133]{Karniel2024ManyValuedML}}]
\label{decidabilitymodal}
The logic KLFI1 is decidable, and so are its extensions with axioms
K, D, T, 4, B, and 5.
\end{proposition}

In our contribution the following epistemic modal extensions of LFI1 play a role.

\begin{definition} The logic S5LFI1 is the extension of KLFI1 with the axioms $\mathrm{T}$, $\mathrm{4}$, and $\mathrm{5}$. The logic KB4LFI1 is the extension of KLFI1 with the axioms $\mathrm{B}$ and $4$.
\end{definition}

 The logic KLFI1 is interpreted on Kripke models with arbitrary accessibility relations, S5LFI1 on models with equivalence relations, and KB4LFI1 on models with partial equivalence relations (symmetric and transitive). 
 The logic KB4 has been considered for multi-agent systems in distributed computing with agent that may be faulty, either malicious in the sense that they can lie, or merely byzantine \cite{hvdetal.aiml:2022}. Decidability and compactness also hold for S5LFI1.

Paraconsistent modal logics can add contradictory states to Kripke models, thereby affecting knowledge attributions. In informational contexts where agents may receive conflicting information from outside sources, these logics are useful to model the agents' informational perspective. 

\subsection{Example}
We finish the section with an example involving knowledge and paraconsistency.

\begin{example}
    Anne wants to know the color of her eyes. To determine that, she decides to ask each friend that question. Her eyes are light, but she is not sure whether they are blue. Even though she is sure whether some of her friends will tell her that the color of her eyes are blue, she is not whether all of them will agree on that. To simplify this setting, consider $p_b$: ``Anne's eyes are blue'', and consider the epistemic model $\mathcal{B}$ depicted below. Given
paraconsistency, a state is depicted by $p_b$, in case $p_b$ is true in this state, or $\dot{p}_b$, in case $p_b$ is contradictory in this state. An arrow $\stackrel a \imp$ between states represents a pair in the accessibility relation $\sim_i$.
    In our current case, Anne's uncertainty is between $p_b$ and $\dot{p}_b$, that is, whether there will be agreement or conflict of information between her friends' answers. 
\end{example}

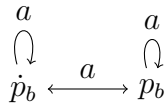
\begin{tikzpicture}
\node (01) at (1.7,0) {$\dot{p}_b$};
\node (1) at (3.4,0) {$p_b$};
\draw[<->] (01) to node[above] {$a$} (1);
\draw[->] (01) edge[loop above,looseness=9] node[above] {$a$} (01); 
\draw[->] (1) edge[loop above,looseness=9] node[above] {$a$} (1); 
\end{tikzpicture}

In the example, $\mathcal{B} \vDash K_a p_b$, however, $\mathcal{B} \nvDash K_a(\circ p_b \lor \bullet p_b)$. That is, Anne knows she has positive information about $p_b$, however, she is not sure whether she has negative information. If she is in the state $\dot{p}_b$, then she does have conflicting information about $p_b$ (although she does not know this): $\vartheta_{\dot{p}_b}(p_b \land \neg p_b) = \half$.

\section{Action model paraconsistent logic AMLFI1}\label{section3}

We now introduce the logic AMLFI1, a dynamic extension of S5LFI1. We will later see that it (or rather, a similar logic by a different name) equally extends KB4LFI1 and KLFI1.

\subsection{Syntax and semantics}

\weg{\begin{definition}
Let $A$ be a set of agents such that $i \in N$, and Let $P$ be a set of atomic variables such that $p \in P$. The language $\mathcal{L}_{AM}$ is the language of action model logic defined by
\[
\begin{aligned}
\alpha &:= p \mid \neg\alpha \mid \bullet\alpha \mid (\alpha \land \alpha)
     \mid K_i \alpha \mid [U,e]\alpha 
\end{aligned}
\]
\end{definition}
}
Consider the language $\mathcal{L}_{AM}$ defined in Definition~\ref{language}. For $[e]\alpha$ we write $[U,e]\alpha$ where $(U,e)$ is a pointed action model simultaneously defined below.

\begin{definition}
   The logic AMLFI1 is characterized over the language $\mathcal{L}_{AM}$ by the axioms and rules of S5LFI1 plus the following axioms and rules:
\[
\begin{array}{l@{\qquad}l}
\begin{array}{l l}
AM1 & [U,e]p \equiv (\pre(e) \rightarrow p)\\
AM2 & [U,e]\neg\alpha \equiv (\pre(e) \rightarrow \neg[U,e]\alpha)\\
AM3 & [U,e]\bullet\alpha \equiv (\pre(e) \rightarrow \bullet[U,e]\alpha)
\end{array}
&
\begin{array}{l l}
AM4 & [U,e](\alpha \land \beta)
      \equiv ([U,e]\alpha \land [U,e]\beta)\\
AM5 & [U,e]K_i\alpha
      \equiv (\pre(e) \rightarrow \bigwedge_{e \sim_i f} K_i[U,f]\alpha)\\
RE  & \theta \equiv \gamma \Rightarrow\ \alpha \equiv \alpha(\theta/\gamma)
\end{array}
\end{array}
\]
\end{definition}
We recall that $\equiv$ is strong equivalence (Definition \ref{definability}). Equivalence $\eq$ is insufficient to get a complete axiomatization.

\begin{definition}[Action Model]
An S5LFI1 {\em action model} $U$ is a structure $\langle \mathit{E}, \{\sim_i\}_{i\in N}, \mathrm{pre}\rangle$, such that $E$ is a set of \emph{actions}, $\sim_i$ is an equivalence relation for an agent $i$, and {\em precondition function} $\mathrm{pre}:\mathit{E} \mapsto \mathcal{L}_{AM}$ assigns a formula to each action. The formula $\mathrm{pre}(e)$ is the \emph{precondition} for the execution of the action $e$.
\end{definition}
Given an action model $U = \langle E, \{\sim_i\}_{i \in N}, \mathrm{pre}\rangle$, we abbreviate $\displaystyle\bigwedge_{e\in E} [U,e]\alpha$ as $[U]\alpha$.

Public announcements are action models with a domain of cardinality $1$, where the precondition of the action is the publicly announced formula, accessible by all agents. Given such a pointed action model $(U,e)$, with $\mathrm{pre}(e)=\beta$, we write $[\beta]\alpha$ for $[U,e]\alpha$
\cite{van2008dynamic}.

We can now explain the countably infinite index set $I$ of which the elements $e \in I$ label the dynamic modalities $[e]\alpha$, that is, $[U,e]\alpha$. It is an enumeration of all finite pointed action model frames, and wherein all actions (points) get different names. As the set of agents is finite and as the domains of the action models are finite, such an enumeration exists.  However, we should then not see $[U,e]\alpha$ as featuring in the BNF language definition as a modality with a single argument but let the preconditions of actions count as additional formula arguments. For example, a public announcement formula $[\alpha]\beta$ can be seen as a modality with two arguments $\alpha$ and $\beta$ (in fact, this is how Plaza originally formalized the public announcement \cite{plaza:1989}). We should therefore rather have written $[e]\vec{\alpha}$ than $[e]\alpha$ in Definition~\ref{language}, where $\vec{\alpha}$ is a $(n+1)$-tuple of formulas $\alpha_1,\dots,\alpha_n,\alpha$ given that the domain of $(U,e)$ has cardinality $n$. (The enumeration argument extends to the action models in the next section that also contain postconditions --- for a finite number of atoms. We just make it an even longer enumeration. The Hilbert Hotel can still accommodate us.)

We should already note that the logic AMLFI1 is not a normal modal logic, as common for dynamic epistemic logics. For a counterexample, $[p]p$ is a theorem but $[p \et \neg K_i p](p \et \neg K_i p)$ is not a theorem. The shape of the above axioms makes clear that uniform substitution cannot hold: some axioms feature atoms, whereas other axioms feature arbitrary formulas; we therefore cannot simply replace $p$ by $\neg \alpha$ or by $\alpha\et\beta$ in some theorem that features the atom after an announcement.

\begin{definition}[Semantics for AMLFI1]
     Given a Kripke model $M = \langle S, \{\sim_i\}_{i \in N}, \vartheta\rangle$, and action model $U = \langle\mathrm{E}, \{\sim_i\}_{i \in N}, \mathrm{pre}\rangle$, {\em update} $(M \otimes U, (s, e)) =  \langle S', \sim', \vartheta'\rangle$ (a restricted modal product) is defined as:
\[\begin{array}{l}
     S' = \{(s, e): s \in S, e \in E, \text{ and } \vartheta_s(\mathrm{pre}(e)) \in D\}  \\
          (s, e) \sim_i' (t, f) \text{ iff } s \sim_i t \text{ and } e \sim_i f \\
              \vartheta'_{(s,e)}(p) = \vartheta_s(p)
\end{array}\]
We then extend the semantic definition of an S5LFI1-valuation with the additional condition:
\[
\begin{array}{l}
  \vartheta_s([U,e]\alpha) = \begin{cases}
\vartheta'_{(s,e)}(\alpha) &\text{if } \vartheta_s(\mathrm{pre}(e)) \in D  \\
1 &\text{otherwise}
\end{cases}
\end{array}\]
\end{definition}

A valuation that satisfies these conditions is called an AMLFI1-valuation. The logical consequence $\Gamma \vDash_{AMLFI1} \alpha$ holds iff for all AMLFI1 models $\mathcal{M}$ and all AMLFI1-valuations $\vartheta_w$ such that $w \in W$ and $\vartheta_w \in \vartheta$, $\vartheta_w[\Gamma] \subseteq \mathrm{D}$ implies that $\vartheta_w(\alpha) \in \mathrm{D}$.

\subsection{Soundness and completeness of AMLFI1}

We first show soundness.

\begin{theorem}[Soundness] \label{soundness} If $\Gamma \vdash \alpha$, then $\Gamma \vDash \alpha$.
\end{theorem}

\begin{proof}
{\bf AM1} $\bm{[U,e]p \equiv (\mathrm{{\bf pre}}(e) \rightarrow p)}$. Suppose that $\vartheta_s(pre(e)) = 0$. Then $\vartheta_s([U,e]p) = 1$. On the other hand, $\vartheta_s(\mathrm{pre}(e) \rightarrow p) = 1$. Therefore, $\vartheta_s([U,e]p \equiv (\mathrm{pre}(e) \rightarrow p)) \in \mathrm{D}$. 

Suppose that $\vartheta_s(\mathrm{pre}(e)) \in \mathrm{D}$. Given the conditions of the restricted modal product, $\vartheta_{(s,e)}(p) = \vartheta_s(p)$. Therefore, $\vartheta_s([U,e]p) = \vartheta_s(p)$. On the other hand, given that $\vartheta_s(\mathrm{pre}(e)) \in \mathrm{D}$, $\vartheta_s(\mathrm{pre}(e) \rightarrow p) = \vartheta_s(p)$. Therefore, $\vartheta_s([U,e]p) = \vartheta_s(\mathrm{pre}(e) \rightarrow p)$. Thus, $\vartheta_s([U,e]p \equiv (\mathrm{pre}(e) \rightarrow p)) \in D$. \\
{\bf AM2} $\bm{[U,e]\neg\alpha \equiv (\mathrm{{\bf pre}}(e) \rightarrow \neg[U,e]\alpha)}$. Suppose that $\vartheta_s(\mathrm{pre}(e)) = 0$. In that case, $\vartheta_s([U,e]\neg\alpha) = 1$ and $\vartheta_s(\mathrm{pre}(e) \rightarrow \neg[U,e]\alpha) = 1$. Therefore, $\vartheta_s([U,e]\neg\alpha \equiv (\mathrm{pre}(e) \rightarrow \neg[U,e]\alpha)) \in \mathrm{D}$. 

    Suppose that $\vartheta_s(\mathrm{pre}(e)) \in \mathrm{D}$. 

        \noindent 1. $\vartheta_s([U,e]\neg\alpha) = 1$. In that case, $\vartheta_{(s,e)}(\neg\alpha) = 1$. Therefore, $\vartheta_{(s,e)}(\alpha) = 0$. Therefore, $\vartheta_s([U,e]\alpha) = 0$, which implies that $\vartheta_s(\neg[U,e]\alpha) = 1$. Thus, $\vartheta_s(\mathrm{pre}(e) \rightarrow \neg[U,e]\alpha) = 1$. Therefore, the equivalence holds. \\
2. $\vartheta_s([U,e]\neg\alpha) = \half$. In that case, $\vartheta_{(s,e)}(\neg\alpha) = \half$, which implies that $\vartheta_{(s,e)}(\alpha) = \half$. Thus, $\vartheta_{s}([U,e]\alpha) = \half$, which implies that $\vartheta_{s}(\neg[U,e]\alpha) = \half$, which implies that $\vartheta_s(\mathrm{pre}(e) \rightarrow \neg[U,e]\alpha)=\half$. Therefore, the equivalence holds. \\
        3. $\vartheta_s([U,e]\neg\alpha) = 0$. In that case, $\vartheta_{(s,e)}(\neg\alpha) = 0$, thus, $\vartheta_{(s,e)}(\alpha) = 1$. Therefore, $\vartheta_s([U,e]\alpha) = 1$, which implies that $\vartheta_s(\neg[U,e]\alpha) = 0$, therefore, $\vartheta_s(\mathrm{pre}(e) \rightarrow \neg[U,e]\alpha) = 0$. Thus, the equivalence holds. 
    



    
\noindent \textbf{AM3} $\bm{[U,e]\bullet\alpha \equiv (\mathrm{{\bf pre}}(e) \rightarrow \bullet[U,e]\alpha)}$. If $\vartheta_s(\pre(e)) = 0$, the reasoning is similar to the previous case. 

Suppose that $\vartheta_s(\pre(e)) \in \mathrm{D}$. In that case, \\
 1. $\vartheta_s$($[U,e]\bullet$$\alpha$) = 1. In that case, $\vartheta_{(s,e)}$($\bullet\alpha$) = 1, therefore, $\vartheta_{(s,e)}$($\alpha$) = $\half$. Therefore, $\vartheta_s$($[U,e]\alpha$) = $\half$, which implies that $\vartheta_s$($\bullet[U,e]\alpha$) = 1. Therefore, $\vartheta_s$(pre($e$) $\rightarrow$ $\bullet[U,e]\alpha$) = 1. Therefore, the equivalence holds.\\
2. $\vartheta_s$($[U,e]\bullet$$\alpha$) = 0. In that case, $\vartheta_{(s,e)}$($\bullet\alpha$) = 0, which implies that $\vartheta_{(s,e)}$($\alpha$) = 1 or 0 (we don't need to consider $\vartheta_s([U,e]\bullet\alpha) = \half$, because $\vartheta_s(\bullet\alpha) = 0$ or $1$).\\
        2.1 $\vartheta_{(s,e)}$($\alpha$) = 1. In that case, $\vartheta_s$($[U,e]\alpha$) = 1, therefore, $\vartheta_s$($\bullet[U,e]\alpha$) = 0. Therefore, the equivalence holds.\\
        2.2 $\vartheta_{(s,e)}$($\alpha$) = 0. Reasoning is analogous to the previous case.


\noindent \textbf{AM4} $\bm{[U,e](\alpha \land \beta) \equiv ([U,e]\alpha \land [U,e]\beta)}$. Suppose that $\vartheta_s$(pre($e$)) $\notin$ D. In that case, $\vartheta_s$($[U,e]$($\alpha$ $\land$ $\beta$)) = 1. On the other hand, $\vartheta_s$($[U,e]\alpha$) = 1 and $\vartheta_s$($[U,e]\beta$) = 1, therefore, $\vartheta_s$($[U,e]\alpha$ $\land$ $[U,e]\beta$) = 1. Thus, the equivalence holds. 

Suppose, on the other hand, that $\vartheta_s$(pre($e$)) $\in$ D. \\
    1. $\vartheta_s$($[U,e]$($\alpha$ $\land$ $\beta$)) = 1. In that case, $\vartheta_{(s,e)}$($\alpha$ $\land$ $\beta$) = 1. Thus, $\vartheta_{(s,e)}$($\alpha$) = 1 and $\vartheta_{(s,e)}$($\beta$) = 1. Therefore, $\vartheta_s$($[U,e]\alpha$) = 1 and $\vartheta_s$($[U,e]\beta$) = 1. Therefore, $\vartheta_s$($[U,e]\alpha$ $\land$ $[U,e]\beta$) = 1. Thus, the equivalence holds. 
    \\2. $\vartheta_s$($[U,e]$($\alpha$ $\land$ $\beta$)) = $\half$. In that case, $\vartheta_{(s,e)}$($\alpha$) = $\half$ or $\vartheta_{(s,e)}$($\beta$) = $\half$ and $\vartheta_{(s,e)}$($\alpha$), $\vartheta_{(s,e)}(\beta)$ $\geq$ $\half$. Without loss of generality, suppose that $\vartheta_{(s,e)}$($\alpha$) = $\half$. In that case, $\vartheta_s$($[U,e]\alpha$) = $\half$, and $\vartheta_s$($[U,e]\beta$) $\geq$ $\half$, which implies that $\vartheta_s$($[U,e]\alpha$ $\land$ $[U,e]\beta$) = $\half$. Therefore, the equivalence holds.
    \\3. $\vartheta_s$($[U,e]$($\alpha$ $\land$ $\beta$)) = 0. In that case, $\vartheta_{(s,e)}$($\alpha$ $\land$ $\beta$) = 0. Thus, $\vartheta_{(s,e)}$($\alpha$) = 0 or $\vartheta_{(s,e)}$($\beta$) = 0. Without loss of generality, suppose that $\vartheta_{(s,e)}$($\alpha$) = 0. In that case, $\vartheta_s$($[U,e]\alpha$) = 0, which implies that $\vartheta_s$($[U,e]\alpha$ $\land$ $[U,e]\beta$) = 0. Thus, the equivalence holds.

    \noindent \textbf{AM5} $\bm{[U,e]K_i\alpha \equiv (\mathrm{\bf pre}(e) \rightarrow \bigwedge_{e\sim_i t}K_i[M,t]\alpha)}$. If $\vartheta_s$(pre($e$)) $\notin$ D, the reasoning is similar to the previous cases. 

  Otherwise, suppose that $\vartheta_s$(pre($e$)) $\in$ D.\\
      $\vartheta_s$($[U,e]$$K_i\alpha$) = 1. In that case, $\vartheta_{(s,e)}$($K_i\alpha$) = 1, which means that for all $t$ such that $e \sim_i t$ (given that $s \sim_i s$), $\vartheta_{(s,t)}(\alpha) = 1.$ Thus, $\vartheta_s([U,t]\alpha) = 1$. Suppose $\vartheta_k([U,t]\alpha) \neq 1$, for some $k \sim_i s$. Then, $\vartheta_{(k, t)}(\alpha) \neq 1$. However, given that $k \sim_i t$ and $e \sim_i t$, $\vartheta_{(k,t)}(\alpha) = 1$ (contradiction). Therefore, $\vartheta_s$($\bigwedge_{e \sim_i t}$$K_i[U,t]\alpha$) = 1. 
\\2. $\vartheta_s$($[U,e]$$K_i\alpha$) = $\half$. Similar to the previous case.
\\3. $\vartheta_s$($[U,e]$$K_i\alpha$) = 0. Similar to the previous case.

  \item[] \textbf{RE}
    The proof works via induction on the structure of $\alpha$. Assume $\vDash \theta \equiv \gamma$ throughout the proof.\\
    \noindent\textbf{Base case.} $\alpha$ = $p$. In that case, either $\theta$ = $p$ or $\theta$ $\neq$ $p$. \\
    (i) $\theta$ $\neq$ $p$. If $\theta$ $\neq$ $p$, then $p$$(\theta/\gamma)$ = $p$, therefore, $\vDash$ $p$ $\equiv$ $p$$(\theta/\gamma)$, given that $\equiv$ is reflexive.\\
    (ii) $\theta$ = $p$. In that case, $\vDash$ p $\equiv$ $\gamma$, by hypothesis. Furthermore, $p(\theta/\gamma) = \gamma$. 
    Therefore, $\vDash p \equiv p(\theta/\gamma)$.
    \\\noindent\textbf{Inductive Hypothesis}. For all formulas $\beta$ s.t. $\beta$ is a subformula of $\alpha$, RE holds. 
\\\noindent\textbf{Inductive Case.} Let us consider this proof by cases. 
\\ $\bm{(\alpha = \neg\beta)}$ By I.H., $\vDash$ $\beta$ $\equiv$ $\beta(\theta/\gamma)$. Given that $\equiv$ is a congruence relation, then $\vDash$ $\neg\beta$ $\equiv$ $\neg(\beta(\theta/\gamma))$. Given that $\neg\beta(\theta/\gamma) = \neg(\beta(\theta/\gamma))$, the congruence holds.
\\$\bm{(\alpha=\bullet\beta)}$ Similar to previous case.
\\$\bm{(\alpha=K_i\beta)}$ Similar to the case of negation.
\\$\bm{(\alpha=\beta \land \delta)}$ By I.H., $\vDash$ $\beta$ $\equiv$ $\beta(\theta/\gamma)$ and $\vDash$ $\delta$ $\equiv$ $\delta(\theta/\gamma)$. Therefore, $\vDash$ $\beta$ $\land$ $\delta$ $\equiv$ $\beta(\theta/\gamma)$ $\land$ $\delta(\theta/\gamma)$. Given that $\beta(\theta/\gamma)$ $\land$ $\delta(\theta/\gamma)$ = $(\beta \land \delta)$($\theta/\gamma$), then RE holds for this case.
\\$\bm{(\alpha = [U,e]\delta)}$ By I.H., $\vDash$ $\delta$ $\equiv$ $\delta(\theta /
\gamma)$. Given that $[U,e]$ is an action model (not a formula), the step to $\vDash$ $[U,e]\delta$ $\equiv$ $([U,e]\delta)(\theta /
\gamma)$ is straightforward, because the replacement of formulas concerns only what is in $\delta$.



\end{proof}

(Weak) Completeness is shown by demonstrating that each formula with action model modalities is provably equivalent to a formula without. This requires defining a translation from the full language to its fragment without action model modalities, and a complexity measure (refining the usual subformula complexity, that is insufficient for such purposes) in order to show termination of such rewriting. After showing weak completeness, using compactness, soundness and the deduction theorem, we then show strong completeness from weak completeness. 

\begin{definition}[Translation] \label{trans} The translation $t: \mathcal{L}_{AM}$ $\mapsto$ $\mathcal{L}_K$ is defined as follows:
\[
\begin{array}{l@{\qquad}l}
\begin{array}{l l}
t(p) &= p\\
t(\neg\alpha) &= \neg t(\alpha)\\
t(\bullet\alpha) &= \bullet t(\alpha)\\
t(\alpha \land \beta) &= t(\alpha) \land t(\beta)
\end{array}
&
\begin{array}{l l}
t(K_i\alpha) &= K_i t(\alpha)\\
t([U,e]p) &= t(\mathrm{pre}(e) \rightarrow p)\\
t([U,e]\neg\alpha) &= t(\mathrm{pre}(e) \rightarrow \neg[U,e]\alpha)\\
t([U,e]K_i\alpha) &= t(\mathrm{pre}(e) \rightarrow \bigwedge_{e \sim_i f} K_i[U,f]\alpha)
\end{array}
\end{array}
\]
\end{definition}
\begin{definition}[Complexity]\label{comp}
The complexity c: $\mathcal{L}_{AM}$ $\mapsto$ $\mathds{N}$ is defined as (where $|E|$ is the cardinality of $E = \mathcal D(U)$):
\[
\begin{array}{l@{\qquad}l}
\begin{array}{l l}
c(p) &= 1\\
c(\neg\alpha) &= 1 + c(\alpha)\\
c(\bullet\alpha) &= 1 + c(\alpha)\\
c(\alpha \land \beta)
  &= \max\{c(\alpha), c(\beta)\} + 1
\end{array}
&
\begin{array}{l l}
c(K_i\alpha) &= 1 + c(\alpha)\\
c([U,e]\alpha)
  &= (6 + c([U,e]))\,c(\alpha)\\
c([U,e])
  &= \max\{c(pre(f)) : f \in U\} + |E|
\end{array}
\end{array}
\]
\end{definition}
From these definitions, it is straightforward to check that $c(\alpha \rightarrow \beta) = \max\{c(\alpha)+2, c(\beta)\} + 3$, given that $\alpha \rightarrow \beta$ is definable as $\sim$$\alpha \lor \beta$, and $\alpha \lor \beta$ is definable as $\neg(\neg\alpha \land \neg\beta)$.  
\begin{lemma}\label{complemma}
For all $\alpha, \beta \in \mathcal{L}_{AM}$, propositional variables $p$, and action models [U,e],
\[
\begin{array}{l@{\qquad}l}
\begin{array}{l l}
1. & c(\beta) > c(\alpha), \text{ if } \alpha \in Sub(\beta)\\
2. & c([U,e]p) > c(\mathrm{pre}(e) \rightarrow p)\\
3. & c([U,e]\neg\alpha) > c(\mathrm{pre}(e) \rightarrow \neg[U,e]\alpha)
\end{array}
&
\begin{array}{l l}
4. & c([U,e]\bullet\alpha) > c(\mathrm{pre}(e) \rightarrow \bullet[U,e]\alpha)\\
5. & c([U,e](\alpha \land \beta)) > c([U,e]\alpha \land [U,e]\beta)\\
6. & c([U,e]K_i\alpha) > c(\mathrm{pre}(e) \rightarrow \bigwedge_{e \sim_i f} K_i[U,f]\alpha)
\end{array}
\end{array}
\]
\end{lemma}

\begin{proof}
        \textbf{1}. Follows directly from Definition \ref{comp}.
        \\\textbf{2}. $c([U,e]p) = c([U,e]) + 6 = max\{c(\mathrm{pre}(t)): t \in E\} + 6 = c(\mathrm{pre}(t)) + 6$. On the other hand, $c(\mathrm{pre}(e) \rightarrow p) = \max\{c(\mathrm{pre}(e)) + 2, c(p)\} + 3 = c(\mathrm{pre}(e)) + 5$. Given that $c(pre(e)) \leq c(pre(t))$, then $c([U,e]p) > c(pre(e) \rightarrow p)$.
        \\\textbf{3}. c($[U,e]\neg\alpha$) = (c($[U,e]$) + 6)(c($\alpha$) + 1) = c($[U,e]$c($\alpha$) + c($[U,e]$) + 6c($\alpha$) + 6. On the other hand, c(pre($e$) $\rightarrow$ $\neg[U,e]\alpha$) = max(c(pre($e$)) + 2, c($\neg[U,e]\alpha$)) + 3. Suppose that c($\neg[U,e]\alpha > c(pre($e$)) + 2$. Then, max(c(pre($e$)) + 2, c($\neg[U,e]\alpha$)) + 3 = c($\neg[U,e]\alpha$) + 3 = c($[U,e]\alpha$) + 4 = (c($[U,e]$) + 6)(c($\alpha$)) + 4 = c($[U,e]$)c($\alpha$) + 6c($\alpha$) + 4. 
\\Furthermore, c$([U,e])$c($\alpha$) + c($[U,e]$) + 6c($\alpha$) + 6 = ($\max\{c(pre(f): f \in E\} + |E|)(c(\alpha)) + \max\{c(pre(f): f \in E\} + |E| +6c(\alpha) + 6$ = $c(pre(f))c(\alpha) + |E|(c(\alpha)) + c(pre(f)) + |E| + 6c(\alpha) + 6$. Now, suppose that max(c(pre($e$)) + 2, c($\neg[U,e]\alpha$)) = c(pre(e)) + 2. Then, max(c(pre($e$)) + 2, c($\neg[U,e]\alpha$)) + 3 = c(pre(e)) + 5. Clearly, $c(pre(f))c(\alpha) + |E|(c(\alpha)) + c(pre(f)) + |E| + 6c(\alpha) + 6 > c(pre(e)) + 5$. Therefore, $c([U,e]p) > c(pre(e) \rightarrow p)$. 
\\\textbf{4}. Analogous to item 3.
        \\\textbf{5}. c($[U,e]$($\alpha$ $\land$ $\beta$)) = (c($[U,e]$) + 6)(c($\alpha$ $\land$ $\beta$)) = c($[U,e]$ + 6)(max(c($\alpha$), c($\beta$)) + 1). Suppose, without loss of generality, that $max\{c(\alpha), c(\beta)\} = c(\alpha)$. Then, (c($[U,e]$) + 6)(max\{c($\alpha$), c($\beta$)\} + 1) = (c($[U,e]$) + 6)(c($\alpha$) + 1) = c($[U,e]$)c($\alpha$) + 6c($\alpha$) + c($[U,e]$) + 6. On the other hand, c($[U,e]\alpha$ $\land$ $[U,e]\beta$) = $\max\{c([U,e]\alpha), c([U,e]\beta)\} + 1 = c([U,e]\alpha) + 1 = (c([U,e]) + 6)(c(\alpha)) + 1 = c([U,e])c(\alpha) + 6c(\alpha) + 1$. Therefore, $c([U,e](\alpha \land \beta)) > c([U,e]\alpha \land [U,e]\beta)$.
\\\textbf{6}. $c([U,e]K_i\alpha) = (c([U,e]) + 6)(\alpha + 1) = c([U,e])c(\alpha) + c([U,e]) + 6c(\alpha) + 6$. On the other hand, $c(pre(e) \rightarrow \bigwedge_{e \sim_i f} K_i[U,f]\alpha$ = $\max\{c(\mathrm{pre}(e)), c(\bigwedge_{e \sim_i f} K_i[U,f]\alpha)\} + 3$ = $c(\bigwedge_{e \sim_i f} K_i[U,f]\alpha) + 3$ = \\$\max\{c(K_i[U,f]\alpha):f \in E\} + |E| - 1 + 3 = c(K_i[U,f]\alpha) + |E| + 2 = c([U,f]\alpha) + |E| + 3 = (c([U,f]) + 6)(c(\alpha)) + |E| + 3 = c([U,f])(c(\alpha)) + 6c(\alpha) + |E| + 3$.
      
Now, we need to prove that $c([U,e])c(\alpha) + c([U,e]) + 6c(\alpha) + 6 > c([U,f])(c(\alpha)) + 6c(\alpha) + |E| + 3$. To solve that inequality, we need to prove that $c([U,e]) + 3 > |E|$. $c([U,e]) + 3 = \max\{c(pre(f)):f \in E\} + |E| + 3 > |E|$. Therefore, the inequality holds. 
\end{proof}

\begin{remark}
   Because each rewriting procedure of a formula $\alpha$ to its translation $t(\alpha)$ results in a formula that is less complex than the initial one (except in the final case where $\alpha = p$), then this is a terminating rewriting procedure. That is, translating a given formula in $\mathcal{L}_{AM}$ always terminates. 
\end{remark}

\begin{lemma}\label{completenesslemma}
For all formulas $\alpha$ $\in$ $\mathcal{L}_{AM}$, there is a formula $\beta$ $\in$ $\mathcal{L}_K$ such that $\vdash \alpha \equiv \beta$.

\end{lemma}

\begin{proof} We proceed by induction on the complexity of formulas.\\
    \textbf{Base Case.} $\alpha$ = $p$. In that case, given that t($p$) = $p$, $\vdash$ $p$ $\equiv$ t($p$) ($\equiv$ is a congruence relation, therefore, it is reflexive).\\
   \textbf{Inductive Hypothesis.} Suppose that $\vdash$ $\beta$ $\equiv$ t($\beta$), for all $\beta$ such that c($\beta$) $<$ c($\alpha$).
 \\ \textbf{Inductive Cases.}

            \noindent$\bm{(\alpha = \neg\beta)}$ By I.H., $\vdash$ $\beta$ $\equiv$ t($\beta$). Given that $\equiv$ is a congruence relation, $\vdash$ $\neg\beta$ $\equiv$ $\neg$t($\beta$). By Definition \ref{trans}, t($\neg\beta$) = $\neg$t($\beta$). Therefore, $\vdash$ $\neg\beta$ $\equiv$ t($\neg\beta$). 
            \\\noindent$\bm{(\alpha = \bullet\beta)}$ By I.H., $\vdash$ $\beta$ $\equiv$ t($\beta$). Given that $\equiv$ is a congruence relation, $\vdash$ $\bullet\beta$ $\equiv$ $\bullet$t($\beta$). By Definition \ref{trans}, $\bullet$t($\beta$) = t($\bullet\beta$). Therefore, $\vdash$ $\bullet\beta$ $\equiv$ t($\bullet\beta$).
            \\\noindent$\bm{(\alpha = \beta \land \delta)}$ By induction hypothesis, $\vdash$ $\beta$ $\equiv$ t($\beta$) and $\vdash$ $\delta$ $\equiv$ t($\delta$). Given that $\equiv$ is a congruence rel., $\vdash$ ($\beta$ $\land$ $\delta$) $\equiv$ (t($\beta$) $\land$ t($\delta$)). By Definition \ref{trans}, t($\beta$ $\land$ $\delta$) = t($\beta$) $\land$ t($\delta$). Therefore, $\vdash$ ($\beta$ $\land$ $\delta$) $\equiv$ t($\beta$ $\land$ $\delta$).
            \\\noindent$\bm{(\alpha = K_i\beta)}$ By induction hypothesis, $\vdash$ $\beta$ $\equiv$ t($\beta$). Given that $\equiv$ is a congruence relation, $\vdash$ $\beta$ $\equiv$ t($\beta$) implies $\vdash$ $K_i$$\beta$ $\equiv$ $K_i$t($\beta$). By Definition \ref{trans}, $K_i$t($\beta$) = t($K_i$$\beta$). Thus, $\vdash$ $K_i$$\beta$ $\equiv$ t($K_i$$\beta$).
            \\\noindent$\bm{(\alpha = [U,e]p)}$ By induction hypothesis and Lemma \ref{complemma}.2, $\vdash$ (pre($e$) $\rightarrow$ $p$) $\equiv$ t(pre($e$) $\rightarrow$ $p$). By Definition \ref{trans}, t($[U,e]p$) = t(pre($e$) $\rightarrow$ $p$). Therefore, $\vdash$ (pre($e$) $\rightarrow$ $p$) $\equiv$ t($[U,e]p$). By axiom AM1, $\vdash$ $[U,e]p$ $\equiv$ (pre($e$) $\rightarrow$ $p$). Thus, by transitivity of $\equiv$, $\vdash$ $[U,e]p$ $\equiv$ t($[U,e]p$).
            \\\noindent$\bm{(\alpha = [U,e]\neg\beta)}$ By induction hypothesis, $\vdash$ (pre($e$) $\rightarrow$ $\neg[U,e]\beta$) $\equiv$ t(pre($e$) $\rightarrow$ $\neg[U,e]\beta$). Applying Definition \ref{trans}, $\vdash$ (pre($e$) $\rightarrow$ $\neg[U,e]\beta$) $\equiv$ t($[U,e]\neg\beta$). By axiom AM2, $\vdash$ $[U,e]\neg\beta$ $\equiv$ (pre($e$) $\rightarrow$ $\neg[U,e]\beta$). Thus, by transitivity of $\equiv$, $\vdash$ $[U,e]\neg\beta$ $\equiv$ t($[U,e]\neg\beta$). 
            \\\noindent$\bm{(\alpha = [U,e]\bullet\beta)}$ Case analogous to the previous one using axiom AM3 and Lemma \ref{complemma}.4.
            \\\noindent$\bm{(\alpha = [U,e](\delta \land \gamma))}$ By induction hypothesis, $\vdash$ $[U,e]\delta$ $\equiv$ t($[U,e]\delta$) and $\vdash$ $[U,e]\gamma$ $\equiv$ t($[U,e]\gamma$). Thus, $\vdash$ ($[U,e]\delta$ $\land$ $[U,e]\gamma$) $\equiv$ (t($[U,e]\delta$) $\land$ t($[U,e]\gamma$)), given that $\equiv$ is a congruence relation. By Definition \ref{trans}, t($[U,e]$($\delta$ $\land$ $\gamma$)) = t($[U,e]$$\delta$) $\land$ t($[U,e]$$\gamma$). By axiom AM4, $\vdash$ $[U,e]$($\delta$ $\land$ $\gamma$) $\equiv$ $[U,e]\delta$ $\land$ $[U,e]\gamma$. Therefore, by transitivity of $\equiv$, $\vdash$ $[U,e]$($\delta$ $\land$ $\gamma$) $\equiv$ t($[U,e]$($\delta$ $\land$ $\gamma$)).
            \\\noindent$\bm{(\alpha = K_i\beta)}$ Analogous to the case of negation, using AM5 and Lemma \ref{complemma}.6.
            \\Now, let us consider formulas with more than one dynamic modality. Consider a formula $\alpha$ and its innermost subformula with a dynamic modality $\lambda$ = $[\varphi]\gamma$, $\varphi$ an action model. By i.h, $\vdash$ $[\varphi]\gamma$ $\equiv$ $\beta$, in which $\lambda'$ = $\beta$ $\in$ $\mathcal{L}_K$. Therefore, by RE, $\vdash$ $\alpha$ $\equiv$ $\alpha$($\lambda$/$\lambda'$). Consider the successive replacement of all dynamic modalities in $\alpha$ by its respective translations. By RE, $\vdash$ $\alpha$ $\equiv$ $\alpha$($\lambda_1/\lambda'_1$)($\lambda_2/\lambda'_2$)[...]($\lambda_n/\lambda'_n$). In that case, $\alpha$($\lambda_1/\lambda'_1$)($\lambda_2/\lambda'_2$)[...]($\lambda_n/\lambda'_n$) $\in$ $\mathcal{L}_K$.
\end{proof}

\begin{theorem}[Completeness]\label{completeness}
    For every $\alpha$ $\in$ $\mathcal{L}_{AM}$,
$$\vDash \alpha \text{ implies } \vdash \alpha$$
\end{theorem}

\begin{proof}
    Suppose that $\vDash$ $\alpha$. By Lemma \ref{completenesslemma} and soundness, $\vDash$ $\alpha$ $\equiv$ $\beta$, for some $\beta$ $\in$ $\mathcal{L}$. Therefore, $\vDash$ $\beta$. However, given that $\beta$ $\in$ $\mathcal{L}$, S5LFI1 $\vDash$ $\beta$. By completeness of S5LFI1, S5LFI1 $\vdash$ $\beta$. Given that S5LFI1 is a subsystem of AMLFI1, then $\vdash$ $\beta$. However, $\vdash$ $\alpha$ $\equiv$ $\beta$, therefore, $\vdash$ $\alpha$. 
\end{proof}

\subsection{Compactness, decidability, and strong completeness}

In this subsection, the properties of compactness, decidability, and strong completeness are shown to hold equally well for the action model extensions of S5LFI1.

\begin{corollary}[Compactness]
    AMLFI1 is compact.
\end{corollary}

\begin{proof}
    Suppose $\Gamma \vDash_{AMLFI1} \alpha$. Consider $t[\Gamma] = \{t(\gamma): \gamma \in \Gamma\}.$ Then, $t[\Gamma] \vDash_{AMLFI1} t(\alpha)$, because $\vDash t(\alpha) \equiv \alpha$, for all $\alpha \in \mathcal{L}_{AM}$ (Soundness and Lemma \ref{completenesslemma}). By Compactness of S5LFI1, there exists $\Gamma'$ such that $\Gamma' \subseteq t[\Gamma]$ and $\Gamma'$ is finite such that $\Gamma' \vDash_{S5LFI1} t(\alpha)$. As S5LFI1 is a subsystem of AMLFI1 (see definition of subsystem in Definition~\ref{subsystem}), $\Gamma' \vDash_{AMLFI1}t(\alpha)$. Consider the set of untranslated formulas $\Gamma'^{-1} = \{\alpha: t(\alpha) \in \Gamma'\}$. Given that $\Gamma'$ is finite, so is $\Gamma^{-1}$, and $\Gamma'^{-1}\subseteq \Gamma$. By the strict equivalence of a formula and its translation, it follows that $\Gamma'^{-1} \vDash_{AMLFI1} \alpha$. Therefore, compactness holds for AMLFI1.
\end{proof}

\begin{theorem}[Strong Completeness]
    For every $\Gamma \subseteq \mathcal{L}_{AM}$, $\alpha \in \mathcal{L}_{AM}$,
    \[\Gamma \vDash \alpha \text{ implies } \Gamma \vdash \alpha\]
\end{theorem}

\begin{proof}
    By Soundness (Theorem \ref{soundness}), $\Gamma \cup \{\alpha\} \vDash \beta \text{ implies } \Gamma \vDash \alpha \rightarrow \beta$, for all $\Gamma \subseteq \mathcal{L}_{AM}, \alpha, \beta \in \mathcal{L}_{AM}$. Now, suppose $\Gamma \vDash \alpha.$ By Compactness (Corollary \ref{compactness}), there exists $\Gamma'$ such that $\Gamma' \subseteq \Gamma$ and $\Gamma'$ is finite. Let $\gamma_1, ..., \gamma_n \in \mathcal{L}_{AM}$ such that $\Gamma' = \{\gamma_1, ..., \gamma_n\}$ (since $\Gamma'$ is finite, we can enumerate its formulas like that). In that case, $\Gamma \vDash \alpha$ iff $\vDash (\gamma_1 \land \dots \land \gamma_n) \rightarrow \alpha$. By Completeness (Theorem \ref{completeness}), it follows that $\vdash (\gamma_1 \land \dots \land \gamma_n) \rightarrow \alpha$. By the converse of Deduction Theorem, $\Gamma' \vdash \alpha$. Given that $\Gamma' \subseteq \Gamma$, then, by monotonicity, $\Gamma \vdash \alpha$.
\end{proof}

\begin{theorem}[Decidability]
    AMLFI1 is decidable.
\end{theorem}

\begin{proof}
    It follows from the decidability of KLFI1 in Theorem \ref{decidabilitymodal} and Lemmas \ref{completenesslemma}, \ref{complemma}. Given that there is a terminating rewriting procedure for formulas in $\mathcal{L}_{AM}$ to $\mathcal{L}$, testing the validity of a formula $\alpha$ in AMLFI1 amounts to testing the validity of its translation $t(\alpha) \in \mathcal{L}$, which is a decidable procedure given that S5LFI1 is decidable.
\end{proof}

The soundness and completeness of other action model paraconsistent logics, not for the S5 but for the KB4 extension, with corresponding KB4 (symmetric and transitive) action models, or minimal K normal modal extension of LFI1,  with corresponding action models without restrictions, are virtually the same as for S5LFI1. The important thing to keep in mind is to only consider modal logics interpreted on structures that are closed under action model execution. This is the case for KB4LFI1, KLFI1, and S5LFI1. For a counterexample, consider KDLFI1 interpreted on serial frames. Consider model $M = \langle \{w_1, w_2\}, \sim_i, \vartheta\rangle$ where ${\sim_i} = \{(w_1, w_2), (w_2,w_2)\}$ and $\vartheta_{w_1}(p) = 1$, $\vartheta_{w_2}(p) = 0$. The restriction of $M$ to the state $w_1$ satisfying $p$ is not a serial model anymore.

\subsection{Examples}

We continue with three examples applying the logical semantics. 

\begin{example} \label{examplehone}
Consider agents Anne and Cath. Anne is uncertain whether Cath holds the clubs card (atom $p_c$). Cath now tells Anne she holds clubs. Then, she tells Anne she does not hold clubs. One of these announcements must therefore be a lie. In (standard, truth-valued) public announcement logic we cannot represent that, but in paraconsistent public announcement logic we can. To simplify the setting Cath is an external observer not modelled in this single-agent epistemic model $\mathcal C$ (for `${\mathcal C}$ards') with valuation $\vartheta$.
The public announcement of $p_c$ corresponds to a singleton action model with precondition $p_c$, and similarly for $\neg p_c$. Designated values $1$ and $\half$ are preserved after public announcements. We therefore get the following two updates.

\medskip

\noindent
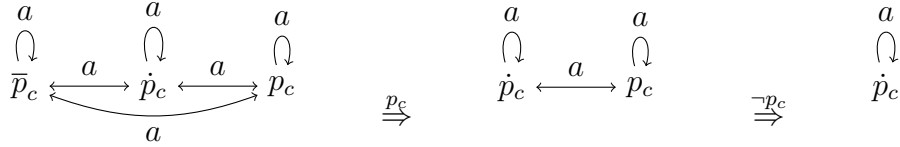
\begin{tikzpicture}
\node (0) at (0,0) {$\ov{p}_c$};
\node (01) at (1.7,0) {$\dot{p}_c$};
\node (1) at (3.4,0) {$p_c$};
\draw[<->] (0) to node[above] {$a$} (01);
\draw[<->] (01) to node[above] {$a$} (1);
\draw[<->,bend right =20] (0) to node[below] {$a$} (1);
\draw[->] (0) edge[loop above,looseness=8] node[above] {$a$} (0); 
\draw[->] (01) edge[loop above,looseness=9] node[above] {$a$} (01); 
\draw[->] (1) edge[loop above,looseness=9] node[above] {$a$} (1); 
\end{tikzpicture}
\qquad \raisebox{10pt}{$\stackrel{p_c}{\Imp}$}\qquad
\raisebox{15pt}{
\begin{tikzpicture}
\node (01) at (1.7,0) {$\dot{p}_c$};
\node (1) at (3.4,0) {$p_c$};
\draw[<->] (01) to node[above] {$a$} (1);
\draw[->] (01) edge[loop above,looseness=8] node[above] {$a$} (01); 
\draw[->] (1) edge[loop above,looseness=9] node[above] {$a$} (1); 
\end{tikzpicture}
}
\qquad \raisebox{10pt}{$\stackrel{\neg p_c}{\Imp}$}\qquad
\raisebox{15pt}{
\begin{tikzpicture}
\node (01) at (1.7,0) {$\dot{p}_c$};
\draw[->] (01) edge[loop above,looseness=8] node[above] {$a$} (01); 
\end{tikzpicture}
}

\medskip

\noindent We now obtain that in model $\mathcal C$, $\vartheta_{\dot{p}_c}([p_c][\neg p_c]K_a \bullet p_c)=1$. Note that this implies that $\vartheta_{\dot{p}_c}([p_c][\neg p_c]K_a(p_c \et\neg p_c))=1$. How to resolve this conflict will later be discussed in Section~\ref{section5}, in Example~\ref{examplehthree}.
\end{example}
\medskip

\begin{example} \label{examplehoneb}
Let us replay the scenario in Example~\ref{examplehone} but now with agent Cath explicitly represented as an agent $c$. We thus get a multi-agent epistemic model $\mathcal C'$ (with valuation $\vartheta'$) with partial equivalence relations, where the relation for agent $c$ is empty in state $\dot{p}_c$, taking into account that agents know their local state. (Why to represent knowledge of local states like that will be discussed later in Section~\ref{section.coup}.) It is shown below. 

In dynamic epistemic logics it is common to represent an agent $c$ saying $\phi$ as a public announcement of $K_c \phi$. If the agent actually knows $\phi$, the announcement is truthful, whereas it is a lying announcement if the agent actually knows $\neg \phi$. We therefore now get the following two updates, according to the semantics of public announcements $K_c p_c$ and after that $K_c \neg p_c$. One of these announcements by $c$ must therefore be a lie. Consider the first announcement. Clearly, $\vartheta'_{p_c}(K_c p_c)=1$ as $p_c$ is the only accessible state for $c$. But note that also $\vartheta'_{\dot{p}_c}(K_c p_c)=1$, as no state is $\sim_c$-accessible from the state named $\dot{p}_c$ and as the infimum of the empty set (given the values $\{0,\half,1\}$ is $1$. We proceed similarly for the second announcement.

We therefore have, similarly to Example~\ref{examplehone}, that $\vartheta'_{\dot{p}_c}([K_c p_c][K_c \neg p_c]K_a \bullet p_c)=1$.

Although this is somewhat satisfactory from $a$'s perspective, who now stores conflicting information for agent $c$, maybe less so from $c$'s perspective, as she is no longer a consistently knowing agent. And agent $a$ might also wish to resolve the conflict concerning $p_c$. We will address these issues in Sections~\ref{section5} and \ref{section.coup}.

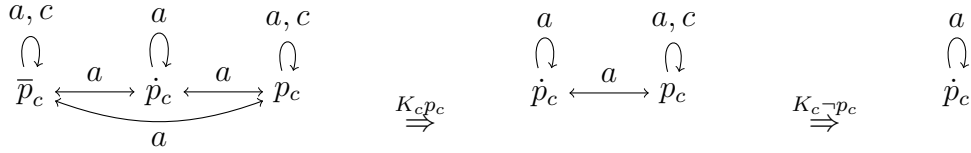
\medskip

\noindent
\begin{tikzpicture}
\node (0) at (0,0) {$\ov{p}_c$};
\node (01) at (1.7,0) {$\dot{p}_c$};
\node (1) at (3.4,0) {$p_c$};
\draw[<->] (0) to node[above] {$a$} (01);
\draw[<->] (01) to node[above] {$a$} (1);
\draw[<->,bend right =20] (0) to node[below] {$a$} (1);
\draw[->] (0) edge[loop above,looseness=8] node[above] {$a,c$} (0); 
\draw[->] (01) edge[loop above,looseness=9] node[above] {$a$} (01); 
\draw[->] (1) edge[loop above,looseness=9] node[above] {$a,c$} (1); 
\end{tikzpicture}
\qquad \raisebox{10pt}{$\stackrel{K_c p_c}{\Imp}$}\qquad
\raisebox{15pt}{
\begin{tikzpicture}
\node (01) at (1.7,0) {$\dot{p}_c$};
\node (1) at (3.4,0) {$p_c$};
\draw[<->] (01) to node[above] {$a$} (1);
\draw[->] (01) edge[loop above,looseness=8] node[above] {$a$} (01); 
\draw[->] (1) edge[loop above,looseness=9] node[above] {$a,c$} (1); 
\end{tikzpicture}
}
\qquad \raisebox{10pt}{$\stackrel{K_c \neg p_c}{\Imp}$}\qquad
\raisebox{15pt}{
\begin{tikzpicture}
\node (01) at (1.7,0) {$\dot{p}_c$};
\draw[->] (01) edge[loop above,looseness=8] node[above] {$a$} (01); 
\end{tikzpicture}
}
\end{example}

\begin{example} \label{examplehtwo}
Let now two agents Anne and Bill be uncertain about the truth of $p_c$, as in the model $\mathcal C^\sharp$ below. First, (external observer) Cath tells Anne and Bill that she does not hold clubs, a public announcement of $\neg p_c$. Then, she semi-publicly informs Bill but not Anne that her announcement was unreliable because in fact she holds clubs, where semi-public means that (Anne and Bill have joint awareness that) Anne is aware of the action. For example, before whispering this information in Bill's ear she tells Anne and Bill that she will privately inform Bill whether her announcement was reliable. This action corresponds to an action model $\mathcal U$ consisting of two actions (action points) $r$ and $u$ with therefore preconditions $\circ p_c$ (reliable) and $\bullet p_c$ (unreliable), with designated action $u$:

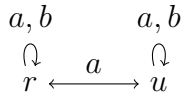
\medskip
\noindent

\begin{tikzpicture}
\node (0) at (0,0) {$r$};
\node (01) at (1.7,0) {$u$};
\draw[<->] (0) to node[above] {$a$} (01);
\draw[->] (0) edge[loop above,looseness=8] node[above] {$a,b$} (0); 
\draw[->] (01) edge[loop above,looseness=8] node[above] {$a,b$} (01); 
\end{tikzpicture}

\medskip

The updates are as follows. Formally, state $\ov{p}_c$ now becomes $(\ov{p}_c,r)$ (as $\vartheta_{\ov{p}_c}(\circ p_c) = 1$) and state $\dot{p}_c$ now becomes $(\dot{p}_c,u)$ (as $\vartheta_{\dot{p}_c}(\bullet p_c) = 1$), but we again represent the resulting states by their valuation, so we get $\ov{p}_c$ respectively $\dot{p}_c$.

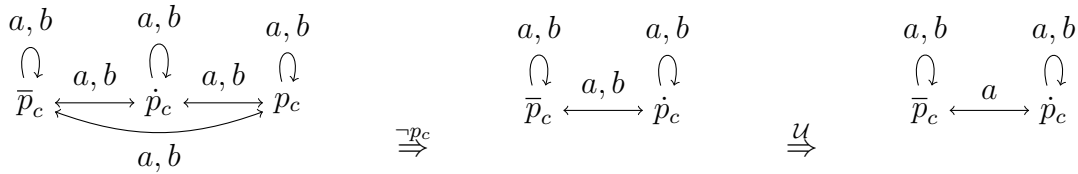
\medskip

\noindent
\begin{tikzpicture}
\node (0) at (0,0) {$\ov{p}_c$};
\node (01) at (1.7,0) {$\dot{p}_c$};
\node (1) at (3.4,0) {$p_c$};
\draw[<->] (0) to node[above] {$a,b$} (01);
\draw[<->] (01) to node[above] {$a,b$} (1);
\draw[<->,bend right =20] (0) to node[below] {$a,b$} (1);
\draw[->] (0) edge[loop above,looseness=8] node[above] {$a,b$} (0); 
\draw[->] (01) edge[loop above,looseness=9] node[above] {$a,b$} (01); 
\draw[->] (1) edge[loop above,looseness=9] node[above] {$a,b$} (1); 
\end{tikzpicture}
\qquad \raisebox{10pt}{$\stackrel{\neg p_c}{\Imp}$}\qquad
\raisebox{18pt}{
\begin{tikzpicture}
\node (0) at (0,0) {$\ov{p}_c$};
\node (01) at (1.7,0) {$\dot{p}_c$};
\draw[<->] (0) to node[above] {$a,b$} (01);
\draw[->] (0) edge[loop above,looseness=8] node[above] {$a,b$} (0); 
\draw[->] (01) edge[loop above,looseness=8] node[above] {$a,b$} (01); 
\end{tikzpicture}
}
\qquad \raisebox{10pt}{$\stackrel{\mathcal U}{\Imp}$}\qquad
\raisebox{18pt}{
\begin{tikzpicture}
\node (0) at (0,0) {$\ov{p}_c$};
\node (01) at (1.7,0) {$\dot{p}_c$};
\draw[<->] (0) to node[above] {$a$} (01);
\draw[->] (0) edge[loop above,looseness=8] node[above] {$a,b$} (0); 
\draw[->] (01) edge[loop above,looseness=8] node[above] {$a,b$} (01); 
\end{tikzpicture}
}

\medskip

\noindent
Then $\mathcal C^\sharp \vDash [\neg p_c] [\mathcal U] (\neg K_a \circ p_c \land \neg K_a \bullet p_c)$: it is valid on model $\mathcal C^\sharp$ that Anne does not know whether Cath's announcement was reliable. However, after the same two updates, $K_b \circ p_c \lor K_b \bullet p_c$ is now valid on $C^\sharp$ --- Bill knows whether Cath's announcement was reliable, and also $K_a (K_b \circ p_c \lor K_b \bullet p_c)$ is valid on $C^\sharp$ --- Anne knows that Bill knows that.
\end{example}

With action models, we can represent paraconsistent updates wherein an agent is informed of the consistency or inconsistency of a proposition privately. Other agents may have different information on this proposition. With this information she can remain consistent and therefore still reason classically about this proposition. Such contradictions can arise when combining databases from different agents, by such exchanges of information. Dynamics formalized with action models permit that such an inconsistency is local in the sense that a contradictory pair of propositions is true without the agent knowing every formula, that is, being globally inconsistent (`getting crazy'). 

\section{Public announcement paraconsistent logic}\label{section4}

A special case of the logic AMLFI1 is the logic PALFI1. This results in the following axioms involving public announcements, where we also feature the same RE again for reasons to become clear soon. Soundness and completeness follow immediately for this language fragment of AMLFI1.
\[
\begin{array}{l@{\qquad}l}
\begin{array}{l l}
PA1 & [\alpha]p \equiv (\alpha \rightarrow p)\\
PA2 & [\alpha]\neg\beta \equiv (\alpha \rightarrow \neg[\alpha]\beta)\\
PA3 & [\alpha]\bullet\beta \equiv (\alpha \rightarrow \bullet[\alpha]\beta)
\end{array}
&
\begin{array}{l l}
PA4 & [\alpha](\beta \land \gamma) \equiv ([\alpha]\beta \land [\alpha]\gamma)\\
PA5 & [\alpha]K_i\beta \equiv (\alpha \rightarrow K_i[\alpha]\beta)\\
RE  & \theta \equiv \gamma \Rightarrow \alpha \equiv \alpha(\theta/\gamma)
\end{array}
\end{array}
\]
Interestingly, this axiomatization defines the same logic (consists of the same set of theorems) as the recently published \cite{ongaratto2025dynamics} alternative axiomatization for a public announcement extension of paraconsistent LFI1 here denoted PALFI1$_C$. It has a composition of announcements axiom $[\alpha][\beta]\gamma \equiv [\alpha \land [\alpha]\beta]\gamma$, but no derivation rule RE. (Note that necessitation for public announcement is derivable in both axiomatizations.) The axiomatization PALFI1$_C$ in \cite{ongaratto2025dynamics} uses an outside-in reduction strategy to show that all formulas with public announcements are equivalent to formulas without public announcement. Whereas the axiomatization of PALFI1 presented here, as it is the special case of AMLFI1 for public announcements, has a derivation rule RE for replacements of equivalents, and no axiom for composition of announcements. How such different axiomatizations for public announcements correspond has been shown in \cite{WC13}, and in this section we show that these results carry over to the paraconsistent context. We prove that RE is admissible in PALFI1$_C$ (for all formulas, if $\theta$ $\equiv$ $\gamma$ is a theorem in PALFI1$_C$, then $\alpha$ $\equiv$ $\alpha$($\theta/\gamma$) is a theorem in PALFI1$_C$) and that the composition axiom is admissible in PALFI1, thus showing that the logics are the same. These proofs (Lemmas~\ref{composition} and \ref{admissibility of comp}) require another Lemma \ref{properties}. The proofs of all these lemmas are by formula induction.

\begin{lemma}\label{properties} The following properties hold for PALFI1:
\[\begin{array}{llllll}
    1. &  \vdash \neg[\alpha]\beta \rightarrow \neg\sim\alpha &
    \qquad 2. &  \vdash [\alpha]\beta \equiv (\alpha \rightarrow [\alpha]\beta) &
    \qquad 3. & \vdash [\alpha](\beta \rightarrow \gamma) \equiv ([\alpha]\beta \rightarrow [\alpha]\gamma) 
\end{array}\]
\end{lemma}

\begin{proof}

   \textbf{(1)} The proof works by induction on $\beta$.
\\\noindent\textbf{Base Case.} \\\noindent 1. $\vdash$ $\neg[\alpha]p \equiv \neg(\alpha \rightarrow p)$ [PA1, RE]
\\\noindent 2. $\vdash$ $\neg(\alpha \rightarrow p) \equiv \neg$$\sim$$\alpha \land \neg p$ [Def. of $\rightarrow$]
\\\noindent 3. $\vdash$ $\neg[\alpha]p$ $\rightarrow$ $\neg$$\sim$$\alpha$ [Transitivity of $\equiv$, A4]
\\\noindent\textbf{Inductive Hypothesis.} Let $\gamma$ be such that $\gamma$ is a subformula of $\beta$. Then, $\vdash$ $\neg[\alpha]\gamma$ $\rightarrow$ $\neg$$\sim$$\gamma$.
\\\noindent\textbf{Inductive Case.} 
\\\noindent$\bm{(\beta = \neg\gamma)}$ 
\\\noindent 1. $\vdash$ $\neg[\alpha]\neg\gamma$ $\equiv$ $\neg(\alpha \rightarrow \neg[\alpha]\gamma)$ [PA2, RE]
\\\noindent 2. $\vdash$ $\neg(\alpha \rightarrow \neg[\alpha]\gamma) \equiv \neg$$\sim$$\alpha \land [\alpha]\gamma$ [RE, def. of $\rightarrow$]
\\\noindent 3. $\vdash$ $\neg$$[\alpha]\neg\gamma$ $\rightarrow \neg$$\sim$$\alpha$ [A4, transitivity of $\rightarrow$]
\\\noindent$\bm{(\beta = \bullet\gamma)}$ Similar to the case of negation.
\\\noindent$\bm{(\beta = K_i\gamma)}$ Similar to the case of negation.
\\\noindent$\bm{(\beta = \lambda \land \gamma)}$ 
\\\noindent 1. $\vdash$ $\neg[\alpha](\lambda \land \gamma) \equiv \neg([\alpha]\lambda \land [\alpha]\gamma)$ [PA4, RE]
\\\noindent 2. $\vdash$ $\neg([\alpha]\lambda \land [\alpha]\gamma)$ $\equiv$ $\neg[\alpha]\lambda \lor \neg[\alpha]\gamma$ [De Morgan Law]
\\\noindent 3. $\vdash$ $\neg[\alpha]\lambda \rightarrow \neg$$\sim$$\alpha$ [I.H.]
\\\noindent 4. $\vdash$ $\neg[\alpha]\gamma \rightarrow \neg$$\sim$$\alpha$ [I.H.]
\\\noindent 5. $\vdash$ $\neg[\alpha](\lambda \land \gamma)$ $\rightarrow$ $\neg$$\sim$$\alpha$ [2, 3, 4, A8]

\noindent\textbf{(2)} The elementary proof is left to the reader.
\\
    \textbf{(3)} \noindent 

\noindent 1. $\vdash [\alpha](\beta \rightarrow \gamma) \equiv [\alpha]\neg(\neg\sim\beta \land \neg\gamma)$ [Def. of $\rightarrow$, RE]
\\2. $\vdash [\alpha]\neg(\neg\sim\beta \land \neg\gamma) \equiv \alpha \rightarrow \neg[\alpha](\neg\sim\beta \land \neg\gamma)$ [PA2]
\\3. $\vdash \alpha \rightarrow \neg[\alpha](\neg\sim\beta \land \neg\gamma) \equiv \alpha \rightarrow \neg([\alpha]\neg\sim\beta \land [\alpha]\neg\gamma)$ [PA4, RE]
\\4. $\vdash \alpha \rightarrow \neg([\alpha]\neg\sim\beta \land [\alpha]\neg\gamma) \equiv \alpha \rightarrow (\neg[\alpha]\neg\sim\beta \lor \neg[\alpha]\neg\gamma)$ [De Morgan Law, RE] 
\\5. $\vdash \alpha \rightarrow (\neg[\alpha]\neg\sim\beta \lor \neg[\alpha]\neg\gamma) \equiv \alpha \rightarrow (\neg(\alpha \rightarrow \neg[\alpha]\sim\beta) \lor (\neg(\alpha \rightarrow \neg[\alpha]\gamma)))$ [PA2, RE]
\\6. $\vdash \alpha \rightarrow (\neg(\alpha \rightarrow \neg[\alpha]\sim\beta) \lor (\neg(\alpha \rightarrow \neg[\alpha]\gamma))) \equiv \alpha \rightarrow (\neg(\sim\alpha \lor \neg[\alpha]\sim\beta) \lor \neg(\sim\alpha \lor \neg[\alpha]\gamma)$ [Definition of $\rightarrow$]
\\7. $\vdash  \alpha \rightarrow (\neg(\sim\alpha \lor \neg[\alpha]\sim\beta) \lor \neg(\sim\alpha \lor \neg[\alpha]\gamma)) \equiv \alpha \rightarrow ((\neg\sim\alpha \land [\alpha]\sim\beta) \lor (\neg\sim\alpha \land [\alpha]\gamma))$ [De Morgan Law, RE]
\\8. $\vdash \alpha \rightarrow ((\neg\sim\alpha \land [\alpha]\sim\beta) \lor (\neg\sim\alpha \land [\alpha]\gamma)) \equiv \alpha \rightarrow ((\neg\sim\alpha \land (\alpha \rightarrow \sim[\alpha]\beta)) \lor (\neg\sim\alpha \land [\alpha]\gamma))$ [Classical equivalence for DEL, RE]
\\9. $\vdash \alpha \rightarrow ((\neg\sim\alpha \land (\alpha \rightarrow \sim[\alpha]\beta)) \lor (\neg\sim\alpha \land [\alpha]\gamma)) \equiv \alpha \rightarrow ((\neg\sim\alpha \land (\sim \alpha \lor \sim[\alpha]\beta)) \lor (\neg\sim\alpha \land [\alpha]\gamma))$ [Definition of $\rightarrow$] 
\\10. $\vdash \alpha \rightarrow ((\neg\sim\alpha \land (\sim \alpha \lor \sim[\alpha]\beta)) \lor (\neg\sim\alpha \land [\alpha]\gamma)) \equiv \alpha \rightarrow ((\neg\sim\alpha \land \sim[\alpha]\beta) \lor (\neg\sim\alpha \land [\alpha]\gamma))$ [Distributivity of $\land$, RE]
\\11. $\vdash \alpha \rightarrow ((\neg\sim\alpha \land \sim[\alpha]\beta) \lor (\neg\sim\alpha \land [\alpha]\gamma)) \equiv \alpha \rightarrow ((\neg\sim\alpha \land (\sim[\alpha]\beta \lor [\alpha]\gamma))$ [Distributivity of $\land$, RE]
\\12. $\vdash \alpha \rightarrow ((\neg\sim\alpha \land (\sim[\alpha]\beta \lor [\alpha]\gamma)) \equiv \sim \alpha \lor (\neg\sim\alpha \land ([\alpha]\beta \rightarrow [\alpha]\gamma))$ [Definition of $\rightarrow$]
\\13. $\vdash \sim \alpha \lor (\neg\sim\alpha \land ([\alpha]\beta \rightarrow [\alpha]\gamma) \equiv (\sim \alpha \lor \neg\sim\alpha) \land (\sim \alpha \lor ([\alpha]\beta \rightarrow [\alpha]\gamma))$ [Distributivity of $\lor$]
\\14. $\vdash (\sim \alpha \lor \neg\sim\alpha) \land (\sim \alpha \lor ([\alpha]\beta \rightarrow [\alpha]\gamma)) \equiv \sim\alpha \lor ([\alpha]\beta \rightarrow [\alpha]\gamma)$ [A9, RE]
\\15. $\vdash \sim\alpha \lor ([\alpha]\beta \rightarrow [\alpha]\gamma) \equiv \alpha \rightarrow ([\alpha]\beta \rightarrow [\alpha]\gamma)$ [Definition of $\rightarrow$]
\\16. $[\alpha](\beta \rightarrow \gamma) \equiv \alpha \rightarrow [\alpha](\beta \rightarrow \gamma)$ [Lemma \ref{properties}.2]
\\17. $\alpha \rightarrow [\alpha](\beta \rightarrow \gamma) \equiv \alpha \rightarrow ([\alpha]\beta \rightarrow [\alpha]\gamma)$ [Transitivity of $\equiv$, 15, 16]
\\18. $[\alpha](\beta \rightarrow \gamma) \equiv [\alpha]\beta \rightarrow [\alpha]\gamma$ [Congruence property of $\equiv$, 17]
\end{proof}


\begin{lemma}\label{composition}
$[\alpha][\beta]\gamma$ $\equiv$ $[\alpha \land [\alpha]\beta]\gamma$ is admissible in PALFI1.
\end{lemma}

\begin{proof}
    The proofs works by induction on the structure of $\gamma$.

    \noindent\textbf{Base Case.} $\gamma = p$ 
    
    \noindent1. $\vdash$ $[\alpha][\beta]p$ $\equiv$ $[\alpha](\beta \rightarrow p)$ [RE, PA1]
   \\2. $\vdash$ $[\alpha]$($\beta \rightarrow p$) $\equiv$ $[\alpha]\beta \rightarrow [\alpha]p$
 [Lemma \ref{properties}.3]
 \\3. $\vdash$ $[\alpha]\beta \rightarrow [\alpha]p$ $\equiv$ $[\alpha]\beta \rightarrow (\alpha \rightarrow p)$ [RE, PA1]
 \\4. $\vdash$ $[\alpha]\beta \rightarrow (\alpha \rightarrow p)$ $\equiv$ $\alpha \rightarrow ([\alpha]\beta \rightarrow p)$ [RE, Lemma \ref{equivalences}.3] 
\\5. $\vdash$ $\alpha \rightarrow ([\alpha]\beta \rightarrow p)$ $\equiv$ ($\alpha$ $\land$ $[\alpha]\beta$) $\rightarrow p$ [Lemma \ref{equivalences}.2]
\\6. $\vdash$ $(\alpha \land [\alpha]\beta) \rightarrow p \equiv [\alpha \land [\alpha]\beta]p$ [PA1]
\\7. $\vdash$ $[\alpha][\beta]p$ $\equiv$ $[\alpha \land [\alpha]\beta]p$ [1-6, transitivity of $\equiv$]
\\\noindent\textbf{Inductive Hypothesis.} For all formulas $\delta$ such that c($\delta$) $<$ c($\gamma$), $\vdash$ $[\alpha][\beta]\delta \equiv [\alpha \land [\alpha]\beta]\delta$.
\\\noindent\textbf{Inductive Case.} 
\\$\bm{(\gamma = \neg\delta)}$\\ 1. $\vdash$ $[\alpha][\beta]\neg\gamma$ $\equiv$ $[\alpha](\beta \rightarrow \neg[\beta]\gamma)$ [PA2]
\\2. $\vdash$ $[\alpha](\beta \rightarrow \neg[\beta]\gamma) \equiv ([\alpha]\beta \rightarrow [\alpha]\neg[\beta]\gamma)$ [Lemma \ref{properties}.3]
\\3. $\vdash$ $[\alpha]\neg[\beta]\gamma \equiv (\alpha \rightarrow \neg[\alpha][\beta]\gamma$) [PA2]
\\4. $\vdash$ $[\alpha]\beta \rightarrow (\alpha \rightarrow \neg[\alpha][\beta]\gamma)$ $\equiv$ ($\alpha[\beta] \rightarrow [\alpha]\neg[\beta]\gamma)$ [RE, 2, 3]
\\5. $\vdash$ $[\alpha]\beta$ $\rightarrow$ ($\alpha$ $\rightarrow$ $\neg[\alpha][\beta]\gamma$) $\equiv$ $((\alpha \land [\alpha]\beta) \rightarrow \neg[\alpha][\beta]\gamma$) [Lemma \ref{equivalences}.2]
\\6. $\vdash$ $[\alpha][\beta]\gamma \equiv [\alpha \land [\alpha]\beta]\gamma$ [I.H.]
\\7. $\vdash$ $((\alpha \land [\alpha]\beta) \rightarrow \neg[\alpha][\beta]\gamma$) $\equiv$ $((\alpha \land [\alpha]\beta) \rightarrow \neg[\alpha \land [\alpha]\beta]\gamma$) [RE, 5, 6]
\\8. $\vdash$ $((\alpha \land [\alpha]\beta) \rightarrow \neg[\alpha \land [\alpha]\beta]\gamma$) $\equiv$ $[\alpha \land [\alpha]\beta]\neg\gamma$ [PA2]
\\9. $\vdash$ $[\alpha][\beta]\neg\gamma \equiv [\alpha \land [\alpha]\beta]\gamma$ [1-8, transitivity of $\equiv$]
\\$\bm{(\gamma = K_a\delta)}$ Similar to the case of negation.
\\$\bm{(\gamma = \bullet\delta)}$ Similar to the case of negation. 
\\$\bm{(\gamma = \delta \land \lambda)}$ \\1. $\vdash$ $[\alpha][\beta](\delta \land \lambda)$ $\equiv$ $[\alpha]$($[\beta]\delta \land [\beta]\lambda)$ [PA4, RE]
\\2. $\vdash$ $[\alpha]$($[\beta]\delta \land [\beta]\lambda)$ $\equiv$ $[\alpha][\beta]\delta \land [\alpha][\beta]\lambda$ [PA3]
\\3. $\vdash$ $[\alpha][\beta]\delta \land [\alpha][\beta]\lambda \equiv [\alpha \land [\alpha]\beta]\delta \land [\alpha \land [\alpha]\beta]\lambda$ [I.H., RE]
\\4. $\vdash$ $[\alpha \land [\alpha]\beta]\delta \land [\alpha \land [\alpha]\beta]\lambda$ $\equiv$ $[\alpha \land [\alpha]\beta](\delta \land \lambda)$ [PA4]
\\5. $\vdash$ $[\alpha][\beta](\delta \land \lambda) \equiv [\alpha \land [\alpha]\beta](\delta \land \lambda)$ [1-4, transitivity of $\equiv$]
\\$\bm{(\gamma = [\delta]\lambda)}$ For formulas with a dynamic modality, the procedure is straightforward using axioms PA1-PA5 to reduce $\gamma$ into a formula in which the dynamic modality is not the main operator and the same reasoning in the previous items. 
\end{proof}  

\begin{lemma}\label{admissibility of comp}
$\theta$ $\equiv$ $\gamma$ $\Rightarrow$ $\alpha$ $\equiv$ $\alpha$($\theta/\gamma$) is admissible in PALFI1$_C$.
\end{lemma}

\begin{proof}
We assume $\vdash \theta \equiv \gamma$ throughout the proof.
\textbf{Base Case}. $\bm{(\alpha = p)}$ In that case, 
\\1.1 $\theta \neq p$. Then, $\alpha = \alpha(\theta/\gamma)$.
\\1.2 $\theta = p$. Then, p($\theta/\gamma$) = $\gamma$. Given that $\vdash \theta \equiv \gamma$, then $\vdash p \equiv p(\theta/\gamma)$.
\\\textbf{Inductive Hypothesis.} For all formulas $\beta$ s.t. $\beta$ is a subformula of $\alpha$, RE holds. 
 \\\textbf{1}. $\bm{(\alpha = \neg\beta)}$ By I.H., $\vdash \beta \equiv \beta(\theta/\gamma).$ Therefore, $\vdash \neg\beta \equiv \neg(\beta(\theta/\gamma))$. Given that $ \neg(\beta(\theta/\gamma)) = \neg\beta(\theta/\gamma)$, then $\vdash \neg\beta \equiv \neg\beta(\theta/\gamma)$. 
\\\textbf{2}. $\bm{(\alpha = \bullet\beta)}$ Similar to the case of negation.
\\\textbf{3}. $\bm{(\alpha = K_i\beta)}$ Similar to the case of negation.
    \\\textbf{4}. $\bm{(\alpha = \beta \land \gamma)}$. By I.H., $\vdash \beta \equiv \beta(\theta/\gamma)$ and $\vdash \gamma \equiv \gamma(\theta/\gamma)$ Therefore, $\vdash \beta \land \gamma \equiv \beta(\theta/\gamma) \land \gamma(\theta/\gamma)$. Given that $\beta(\theta/\gamma) \land \gamma(\theta/\gamma)$ = $(\beta \land \gamma)(\theta/\gamma)$, then $\vdash \beta \land \gamma \equiv (\beta \land \gamma)(\theta/\gamma)$. 
    \\\textbf{5}. $\bm{(\alpha =  [\beta]p)}$ By PA1, $\vdash [\beta]p \equiv \beta \rightarrow p$. Furthermore, by I.H., $\vdash (\beta \rightarrow p) \equiv (\beta \rightarrow p)(\theta/\gamma)$. By PA1, $\vdash (\beta \rightarrow p)(\theta/\gamma) \equiv [\beta]p(\theta/\gamma)$. By transitivity of $\equiv$, $\vdash$ $[\beta]p \equiv [\beta]p(\theta/\gamma)$.
    \\\textbf{6}. $\bm{(\alpha = [\beta]\neg\gamma)}$ Analogous reasoning to item 5 using axiom PA2. 
    \\\textbf{7}. $\bm{(\alpha = [\beta](\gamma \land \lambda))}$. Analogous reasoning to item 5 using axiom PA4.
    \\\textbf{8}. $\bm{(\alpha = [\beta]K_i\gamma)}$. Analogous reasoning to item 5 using axiom PA5.
    \\\textbf{9}. $\bm{(\alpha = [\beta]\bullet\gamma)}$. Analogous reasoning to item 5 using axiom PA3.
    \\\textbf{10}. $\bm{(\alpha = [\beta][\gamma]\lambda)}$. By Axiom of Composition, $\vdash [\beta][\gamma]\lambda \equiv [\beta \land [\beta]\gamma]\lambda$. By I.H., $\vdash [\beta \land [\beta]\gamma]\lambda \equiv [\beta \land [\beta]\gamma]\lambda(\theta/\delta)$. By Axiom of Composition, $\vdash [\beta][\gamma]\lambda(\theta/\delta) \equiv [\beta \land [\beta]\gamma]\lambda(\theta/\delta)$. By transitivity of $\equiv$, $\vdash [\beta][\gamma]\lambda \equiv [\beta][\gamma]\lambda(\theta/\delta)$.
\end{proof}

\begin{theorem}
    PALFI1 and PALFI1$_C$ coincide, that is, $\vdash_{PALFI1}$ $\alpha$ iff $\vdash_{PALFI1_C}$ $\alpha$, for every formula $\alpha$. 
\end{theorem}

\begin{proof}
    Straightforward from Lemmas \ref{admissibility of comp} and \ref{composition}.
\end{proof}

\section{UMLFI1}\label{section5}

\subsection{Update Model paraconsistent logic}
There are three different ways to update Kripke models: changing (restricting) their domain, changing (restricting) their relations, and changing the valuation of worlds. And all this, including higher-order uncertainty, that may increase the size of the domain, as when executing action models. Naturally, one expects the same variety for update extensions of modal paraconsistent logics. Public announcement logic and action model logic change the domain and the relations, but not the valuation of facts in a given world, and our paraconsistent AMLFI1 fits this pattern. We will now propose further extension of AMLFI1 called UMLFI1 that also allows factual change (also known as ontic change \cite[p.\ 1]{van2008semantic}) in action models. Such factual change is useful when modelling communicating databases where we encounter contradictory information, because in that way we can resolve local contradictions. Executing an action model with factual change represents resolving the contradiction by consulting a reliable source, and thus becoming consistent again. Our presentation of UMLFI1 closely follows \cite{van2008semantic}.

\begin{definition}[UMLFI1]
    UMLFI1 extends S5LFI1 with the axioms and rules of AMLFI1, except that axiom AM1 is now replaced by a novel axiom UM1: 
\[
UM1\ \ \ \  [U,e]p \equiv (\mathrm{pre}(e) \rightarrow \mathrm{post}(e)(p))
\]\end{definition}


The action model with factual change extends the previous action model with an additional parameter $\mathrm{post}: E \imp P \imp \mathcal{L}$ assigning to each action and each atomic formula {\em from a finite subset of set $P$} (for technical reasons, see \cite[p.\ 4]{van2008semantic} and the explanations on enumerating pointed actions models) a \emph{postcondition}, that is a formula in the language.
\begin{definition}[Update model]
    An update model for a finite set of agents $A$ and language $\mathcal{L}$ is a quadruple $U = \langle$E, \{$\sim_i$\}$_{i \in N}$, pre, post$\rangle$, where E is a finite non-empty set of actions, $\{\sim_i\}_{i \in N}$ is a set of equivalence relations for each agent $i \in N$, $\mathrm{pre}: E \imp \mathcal{L}$ is precondition function, and $\mathrm{post}: E  \imp P \imp \mathcal{L}$ is a {\em postcondition function}.
\end{definition}

\weg{\begin{definition}[Semantics for UMLFI1]
    Given epistemic state (M,s), M = $\langle$S, \{$\sim_i$\}$_{i \in N}$, $\vartheta$$\rangle$, and an update model U = $\langle$E, \{$\sim_i$\}$_{i \in N}$, pre, post$\rangle$, we extend the semantic definition of S5LFI1 with the additional conditions: 
\[
\begin{array}{l}
\vartheta_s([U,e]\alpha) = \begin{cases}
\vartheta_{s'}(\alpha) &\text{if (M,s)\lbrack U, e\rbrack(M',s')}\\
1 &\text{otherwise}
\end{cases}
\end{array}\]
\[\begin{array}{lll}
(M,s)\lbrack U, e\rbrack(M',s') & \text{iff} & \vartheta_s(pre(e)) \in D \text{ and } (M',s') = (M \otimes U, (s, e))
\end{array}\]
\end{definition}}
The semantics for UMLFI1 are similar to that for AMLFI1, except that UMLFI1 considers an update model instead of an action model.

\begin{definition}[Execution]
    Let $M$ be an epistemic model, $s$ $\in$ S, and let $U$ be an update model, with $e$ $\in$ E such that $(M,s) \vDash \mathrm{pre}(e)$. The result of executing $(U,e)$ in $(M,s)$ is the model $(M \otimes U, (s, e)) = \langle S', \{\sim_i'\}_{i \in N}, \vartheta'\rangle$, where
\[\begin{array}{lll}
    S' & = &  \{(t, f): \vartheta_t(pre(f)) \in D\} \\
    (t, f) \sim_i (u, g) & \text{ iff } & (t,f), (u, g) 
    \in S' \text{ and } t \sim_i u, f \sim_i g\\
   \vartheta_{(t,f)}'(p)  & = & \vartheta_{t}(post(f)(p))
\end{array}\]
\end{definition}
Soundness and completeness for UMLFI1 is shown as for AMLFI1, the only difference being that axioms AM1 and UM1 are different. Thus, we only prove soundness for UM1:

\begin{theorem}
$\vDash [U,e]p \equiv (\mathrm{pre}(e) \rightarrow \mathrm{post}(e)(p))$
\end{theorem}
\begin{proof}
If $\vartheta_s$(pre($e$)) $\notin$ D, then it is straightforward. If $\vartheta_s$(pre($e$)) $\in$ D, then $\vartheta_{(s,e)}$($p$) = $\vartheta_s(\mathrm{post}(e)(p))$. On the other hand, given that $\vartheta_s$(pre($e$)) $\in$ D, then $\vartheta_s$(pre($e$) $\rightarrow$ post($e$)($p$)) = $\vartheta_s$(post($e$)($p$)). Therefore, the equivalence holds.
\end{proof}
%
%
%
%
Similarly to the previous section, this also gives us (strongly) complete axiomatizations for paraconsistent update model logics over KB4 and arbitrary K modal extensions of LFI1, with update models in the same frame class as the Kripke models wherein they are executed. Furthermore, compactness and decidability follow from the same procedure as for the case of AMLFI1. 
\subsection{Example}

Again, we finish the section with an example. 

\begin{example} \label{examplehthree}
We continue the analysis of Example~\ref{examplehtwo}. Bill knows he has conflicting information for atom $p_c$. Therefore, after all this hassle, Bill challenges Cath to show her card. She has spades, not clubs. So her initial public announcement was actually correct! This epistemic action is (i) a singleton action model $\mathcal U'$ (let us call the action $e$), with (ii) precondition $\pre(e)= \neg p_c$, and (iii) postcondition $\post(e)(p_c) = \bot$ --- where we recall that the always false proposition $\bot$ is definable in LFI1 as $p_c \et \neg p_c \et \circ p_c$ (Definition~\ref{definability}. This results in the following transition:

\medskip

\noindent
\begin{tikzpicture}
\node (0) at (0,0) {$\ov{p}_c$};
\node (01) at (1.7,0) {$\dot{p}_c$};
\draw[<->] (0) to node[above] {$a$} (01);
\draw[->] (0) edge[loop above,looseness=8] node[above] {$a,b$} (0); 
\draw[->] (01) edge[loop above,looseness=8] node[above] {$a,b$} (01); 
\end{tikzpicture}
\qquad \raisebox{10pt}{$\stackrel{\mathcal U'}{\Imp}$}\qquad
\begin{tikzpicture}
\node (0) at (0,0) {$\ov{p}_c$};
\node (01) at (1.7,0) {$\ov{p}_c$};
\draw[<->] (0) to node[above] {$a$} (01);
\draw[->] (0) edge[loop above,looseness=8] node[above] {$a,b$} (0); 
\draw[->] (01) edge[loop above,looseness=8] node[above] {$a,b$} (01); 
\end{tikzpicture}

\medskip

\noindent

In $\mathcal C''$ (with valuation $\theta''$) we now obtain again that Anne and Bill know that Cath holds clubs, but in a stronger non-defeasible sense, that we can formalize in the paraconsistent language: $\theta''_{\ov{p}_a}(K_a (\neg p \land \circ p) \land K_b(\neg p \land \circ p))=1$, and they also know that the other knows that, in fact, it is now common knowledge between $a$ and $b$. Factual change allows us to resolve conflicts when confronted with temporary contradictions. Again, this becomes more complex if we also explicitly represent the uncertainty of agent $c$ (Section~\ref{section.coup}).
\end{example}

\section{Information, truth, lying, and paraconsistency} \label{section.coup}

\subsection{Coup}

The Coup game \cite{coupgame} inspired this research. In the Coup game, players hold cards, and when asked if they hold a certain card, players can lie about card ownership in order to win the game. However, they can also be challenged by other players about the cards they claim to own. If they then reveal the card in question and are being caught out as a liar, they lose that round of the game. Now in the first place, determining equilibria of imperfect information games is already somewhat challenging, as it depends on decision rules over information sets of indistinguishable states of the game \cite{Aumann1995:BIandCKR}. We will not delve into that. But in the second place, how to model the information change in such settings with lying, is already challenging. Coup is a good example where contradictions and inconsistency are not eternal. There is no Liar Paradox! You expect other players to lie, and to lie yourself too, in order to win the game. You will not get `crazy' by hearing $p$ and then $\neg p$ from another player, you merely conclude she is lying, and may wish to resolve the matter by a challenge. Such a challenge is clearly another kind of dynamics: \emph{saying} that you own a card is different from \emph{showing} the same card. What the dynamic turn in paraconsistency can add to formalizing such settings is a fresh perspective on lies and other byzantine behaviour, and novel ways for agents to recover from lies and become consistent again.

\subsection{Lying and world eliminating updates}

We recall Example~\ref{examplehone}. Agent $c$ is an external observer who may announce either $p_c$ or $\neg p_c$. In standard public announcement logic this cannot both happen, in either order, as only one of the announcements can be truthful. We therefore have:

\medskip

\noindent
\begin{tikzpicture}
\node (0) at (0,0) {$\ov{p}_c$};
\node (1) at (3.4,0) {$p_c$};
\draw[<->] (0) to node[above] {$a$} (1);
\draw[->] (0) edge[loop above,looseness=8] node[above] {$a$} (0); 
\draw[->] (1) edge[loop above,looseness=9] node[above] {$a$} (1); 
\end{tikzpicture}
\qquad \raisebox{10pt}{$\stackrel{p_c}{\Imp}$}\qquad
\begin{tikzpicture}
\node (1) at (3.4,0) {$p_c$};
\draw[->] (1) edge[loop above,looseness=9] node[above] {$a$} (1); 
\end{tikzpicture}
\qquad \raisebox{10pt}{$\stackrel{\neg p_c}{\not\Imp}$}\qquad

\medskip

\noindent But in PALFI1 both announcements can happen. This is a modelling advantage of the dynamic extension of LFI1 versus classical dynamic epistemic logics: we can now represent conflicting information. In the paraconsistent semantics, we do not gather truth, but we gather information. Furthermore, we can also recover from such local inconsistencies with factual change. Note that this can indeed properly be called \emph{ontic change}: agent $c$ is the environment and part of the external `world' and by providing yet novel information this therefore counts as a change of such factual information. This can of course work either way, where we only show the change making $p_c$ true, not the change making $p_c$ false. Note that there is no relation with some initial value of $p_c$ representing whether Cath `originally' held some given card ($p_c$ is true) or some other card ($p_c$ is false). So, such recovery is yet another advantage of dynamic LFI1.

\medskip

\noindent
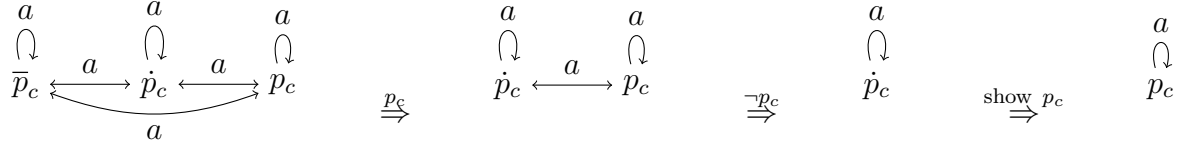
\begin{tikzpicture}
\node (0) at (0,0) {$\ov{p}_c$};
\node (01) at (1.7,0) {$\dot{p}_c$};
\node (1) at (3.4,0) {$p_c$};
\draw[<->] (0) to node[above] {$a$} (01);
\draw[<->] (01) to node[above] {$a$} (1);
\draw[<->,bend right =20] (0) to node[below] {$a$} (1);
\draw[->] (0) edge[loop above,looseness=8] node[above] {$a$} (0); 
\draw[->] (01) edge[loop above,looseness=9] node[above] {$a$} (01); 
\draw[->] (1) edge[loop above,looseness=9] node[above] {$a$} (1); 
\end{tikzpicture}
\qquad \raisebox{10pt}{$\stackrel{p_c}{\Imp}$}\qquad
\raisebox{15pt}{
\begin{tikzpicture}
\node (01) at (1.7,0) {$\dot{p}_c$};
\node (1) at (3.4,0) {$p_c$};
\draw[<->] (01) to node[above] {$a$} (1);
\draw[->] (01) edge[loop above,looseness=8] node[above] {$a$} (01); 
\draw[->] (1) edge[loop above,looseness=9] node[above] {$a$} (1); 
\end{tikzpicture}
}
\qquad \raisebox{10pt}{$\stackrel{\neg p_c}{\Imp}$}\qquad
\raisebox{15pt}{
\begin{tikzpicture}
\node (01) at (1.7,0) {$\dot{p}_c$};
\draw[->] (01) edge[loop above,looseness=8] node[above] {$a$} (01); 
\end{tikzpicture}
}
\qquad \raisebox{10pt}{$\stackrel{\text{show} \ p_c}{\Imp}$}\qquad
\raisebox{15pt}{
\begin{tikzpicture}
\node (01) at (1.7,0) {$p_c$};
\draw[->] (01) edge[loop above,looseness=8] node[above] {$a$} (01); 
\end{tikzpicture}
}

\medskip

\noindent  
Just as in Example~\ref{examplehoneb}, we can replay the same scenario but now with $c$ explicitly represented. There are two ways in which we can envisage to do this, based on different assumptions concerning faulty or absent agents (or in other words, in terms of distributed systems, faulty or crashed processes). We can either assume that agent $c$ considers it possible that $p_c$ is both true and false, in other words, in case $c$ only considers a single state possible, that agent $c$ knows that its own local state is inconsistent, or we assume that agent $c$ knows its local state ($c$ either knows $p_c$ or $c$ knows $\neg p_c$) and thus assume that $c$ does not consider it possible that $p_c$ is true and false simultaneously, in other words: an empty relation $\sim_c$ at such a state. 

The former seems more suitable for byzantine agents (haphazardly spewing forth unreliable information, only appearing to process information) whereas the latter seems more suitable for dead agents (crashed processes, incommunicative agents), and therefore more suitable for the remaining agents to reason about them. This makes the latter more intuitively appealing, except that it requires a more involved update in order to recover from a local inconsistency: we now need not only to make $p_c$ true or false again, representing agent $c$ showing a card to agent $a$, but in the process of doing that we also need to `revive' agent $c$, that is, add a reflexive $c$-arrow to such a state.

The different scenarios are therefore as follows:

\medskip

\noindent
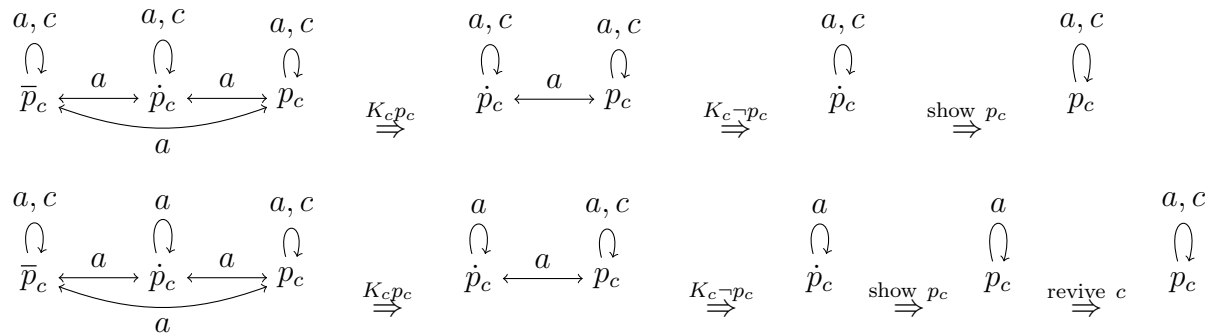
\begin{tikzpicture}
\node (0) at (0,0) {$\ov{p}_c$};
\node (01) at (1.7,0) {$\dot{p}_c$};
\node (1) at (3.4,0) {$p_c$};
\draw[<->] (0) to node[above] {$a$} (01);
\draw[<->] (01) to node[above] {$a$} (1);
\draw[<->,bend right =20] (0) to node[below] {$a$} (1);
\draw[->] (0) edge[loop above,looseness=8] node[above] {$a,c$} (0); 
\draw[->] (01) edge[loop above,looseness=9] node[above] {$a,c$} (01); 
\draw[->] (1) edge[loop above,looseness=9] node[above] {$a,c$} (1); 
\end{tikzpicture}
\quad \raisebox{10pt}{$\stackrel{K_c p_c}{\Imp}$}\quad
\raisebox{15pt}{
\begin{tikzpicture}
\node (01) at (1.7,0) {$\dot{p}_c$};
\node (1) at (3.4,0) {$p_c$};
\draw[<->] (01) to node[above] {$a$} (1);
\draw[->] (01) edge[loop above,looseness=8] node[above] {$a,c$} (01); 
\draw[->] (1) edge[loop above,looseness=9] node[above] {$a,c$} (1); 
\end{tikzpicture}
}
\quad \raisebox{10pt}{$\stackrel{K_c \neg p_c}{\Imp}$}\quad
\raisebox{15pt}{
\begin{tikzpicture}
\node (01) at (1.7,0) {$\dot{p}_c$};
\draw[->] (01) edge[loop above,looseness=8] node[above] {$a,c$} (01); 
\end{tikzpicture}
}
\quad \raisebox{10pt}{$\stackrel{\text{show} \ p_c}{\Imp}$}\quad
\raisebox{15pt}{
\begin{tikzpicture}
\node (01) at (1.7,0) {$p_c$};
\draw[->] (01) edge[loop above,looseness=12] node[above] {$a,c$} (01); 
\end{tikzpicture}
}

\medskip

\noindent
\begin{tikzpicture}
\node (0) at (0,0) {$\ov{p}_c$};
\node (01) at (1.7,0) {$\dot{p}_c$};
\node (1) at (3.4,0) {$p_c$};
\draw[<->] (0) to node[above] {$a$} (01);
\draw[<->] (01) to node[above] {$a$} (1);
\draw[<->,bend right =20] (0) to node[below] {$a$} (1);
\draw[->] (0) edge[loop above,looseness=8] node[above] {$a,c$} (0); 
\draw[->] (01) edge[loop above,looseness=9] node[above] {$a$} (01); 
\draw[->] (1) edge[loop above,looseness=9] node[above] {$a,c$} (1); 
\end{tikzpicture}
\quad \raisebox{10pt}{$\stackrel{K_c p_c}{\Imp}$}\quad
\raisebox{15pt}{
\begin{tikzpicture}
\node (01) at (1.7,0) {$\dot{p}_c$};
\node (1) at (3.4,0) {$p_c$};
\draw[<->] (01) to node[above] {$a$} (1);
\draw[->] (01) edge[loop above,looseness=8] node[above] {$a$} (01); 
\draw[->] (1) edge[loop above,looseness=9] node[above] {$a,c$} (1); 
\end{tikzpicture}
}
\quad \raisebox{10pt}{$\stackrel{K_c \neg p_c}{\Imp}$}\quad
\raisebox{15pt}{
\begin{tikzpicture}
\node (01) at (1.7,0) {$\dot{p}_c$};
\draw[->] (01) edge[loop above,looseness=8] node[above] {$a$} (01); 
\end{tikzpicture}
}
\ \raisebox{10pt}{$\stackrel{\text{show} \ p_c}{\Imp}$} \ 
\raisebox{15pt}{
\begin{tikzpicture}
\node (01) at (1.7,0) {$p_c$};
\draw[->] (01) edge[loop above,looseness=12] node[above] {$a$} (01); 
\end{tikzpicture}
}
\ \raisebox{10pt}{$\stackrel{\text{revive} \ c}{\Imp}$} \
\raisebox{15pt}{
\begin{tikzpicture}
\node (01) at (1.7,0) {$p_c$};
\draw[->] (01) edge[loop above,looseness=12] node[above] {$a,c$} (01); 
\end{tikzpicture}
}

In the first scenario we can imagine $c$ showing $a$ her card ($p_c$ is true), the final epistemic action, as $c$ somehow recovering from her own temporary blackout believing that $p_c$ is both true and false, and then `realizing' while looking at her card again that it was actually the red card ($p_c$) or the green card ($\neg p_c$). She is not lying, but mistaken. So agent $c$ is then more a byzantine agent (a malfunctioning agent, whatever the cause, and consciously or unconsciously, as in distributed computing \cite{KuznetsP0F19,FischerLP85,DH08}) than a lying agent (a conscious, and intentional, action \cite{fallis:2009,hvd.lying:2014,bok:1978}). Still, building upon this scenario, one can even imagine $c$ to be colour-blind, so the card looks both red and green to her when {\em telling} agent $a$ about her card, while when {\em showing} agent $a$, who is not colour-blind, her card, then resolves the issue. That would no longer make $c$ a liar from her own perspective, but not from $a$'s perspective, of course --- unless $a$ was aware of $c$'s colour-blindness. We are running into epi-cycles here.

In the second scenario we need some additional technical machinery to {\em add a loop} for agent $c$, that is, to expand the relation $\sim_c$ of agent $c$ with an additional pair (or, in general, multiple pairs) so that the transition in that way restores or revives agent $c$ as a consistent believer. {\em Expansion of accessibility relations} such as adding loops is what is known in AGM belief revision as \emph{belief contraction} or if you wish \emph{knowledge contraction} \cite{agm:1985} (expanding relations corresponds syntactically to forgetting information, see also the recent \cite{hvd.tark:2025}), and also as `bridging worlds' \cite{ArecesFH15,ArecesFHM18} (adding a pair to an accessibility relation, typically but not necessarily between different worlds for the terminological imagery). We can also imagine the final two actions `show $p_c$' and `revive $c$' as a single epistemic action (combining factual change with relational change), corresponding to what in distributed systems is known as `roll-back' and other ways restoring incorrect agents to their prior correct state \cite{DolevFPSSL14,DitmarschFKS24}. 

Such machinery is not complex, but goes beyond the scope of our current results and seems therefore better left to future research.

We should not fail to mention that in such a second approach we can no longer distinguish a lie from an error, or even random byzantine behaviour. In the $\dot{p}_c$ state, agent $c$'s accessibility relation is empty, which in modal logic corresponds to saying that $c$ is `crazy': $K_c \phi$ is then true for any formula $\phi$, including $\bot$. In other words, $c$ is fully byzantine. In a way, this is not surprising, as LFI1 is a logic designed to track local inconsistencies in communicating databases \cite{Carnielli2000-CARFIA-2}. There is nothing intentional going on here, as in lying, and this is also very similar to the standard procedure in distributed computing; there is no need to call agents crazy, we just call them incorrect or faulty \cite{KuznetsP0F19,FischerLP85,DH08}. 

\subsection{Lying and arrow eliminating updates}

When representing lying of agents as incorrect behaviour we can still stick to partial equivalence relations and updates consisting of, partially observable, model restrictions. But when representing lying of agents as a conscious actions preserving the consistency of their beliefs and the beliefs of other agents, model restrictions no longer suffice and we need relational restrictions. This is because the agent being lied to, and believing the lie, no longer considers the actual world possible: we lose reflexivity of our epistemic models \cite{hvd.lying:2014}. There is a time-honoured update mechanism for that, going back to \cite{gerbrandyetal:1997,gerbrandy:1999}. Let us succinctly also describe that alternative and some of its consequences.

We restart at the beginning. Player Anne is uncertain about player Cath holding a clubs card, atom $p_c$. Cath lying that $p_c$ assumes that she knows that $p_c$ is false. In the resulting model Anne no longer considers the actual state where $p_c$ is false possible. If Cath subsequently announces $\neg p_c$, Anne's accessibility relation is empty: she is crazy. All this results in the following updates (note that the world-eliminating public announcements \cite{plaza:1989} of $p_c$ and of $\neg p_c$ are singleton action models, but that the corresponding arrow-eliminating public announcements \cite{gerbrandyetal:1997} are action models consisting of two actions). We again name the states by their valuations, where $\neg p_c$ is denoted $\ov{p}_c$.

\medskip

\begin{tikzpicture}
\node (0) at (0,0) {$\ov{p}_c$};
\node (1) at (2,0) {$p_c$};
\draw[<->] (0) to node[above] {$a$} (1);
\draw[->] (0) edge[loop above,looseness=8] node[above] {$a$} (0); 
\draw[->] (1) edge[loop above,looseness=9] node[above] {$a$} (1); 
\end{tikzpicture}
\qquad \raisebox{10pt}{$\stackrel{p_c}{\Imp}$}\qquad
\begin{tikzpicture}
\node (0) at (0,0) {$\ov{p}_c$};
\node (1) at (2,0) {$p_c$};
\draw[->] (0) to node[above] {$a$} (1);
\draw[->] (1) edge[loop above,looseness=9] node[above] {$a$} (1); 
\end{tikzpicture}
\qquad \raisebox{10pt}{$\stackrel{\neg p_c}{\Imp}$}\qquad
\begin{tikzpicture}
\node (0) at (0,0) {$\ov{p}_c$};
\node (1) at (2,0) {$p_c$};
\end{tikzpicture}

\medskip

Whether Cath's announcement is a lie or is the truth, depends on the actual state. Given $\neg p_c$, the first announcement is a lie and the second announcement is the truth; whereas given $p_c$, the first announcement would be the truth and the second a lie. After two contradictory announcements Anne is crazy, no matter the actual state. 

In paraconsistency we have addressed this by the expanded setting also representing the state where $p_c$ is both false and true, a designated value $\half$ such that this state is preserved after both announcements. We therefore get the following updates. Again, the state where  $p_c$ is false and true is named after this valuation and denoted $\dot{p}_c$.

\medskip

\noindent
\begin{tikzpicture}
\node (0) at (0,0) {$\ov{p}_c$};
\node (01) at (1.7,0) {$\dot{p}_c$};
\node (1) at (3.4,0) {$p_c$};
\draw[<->] (0) to node[above] {$a$} (01);
\draw[<->] (01) to node[above] {$a$} (1);
\draw[<->,bend right =20] (0) to node[below] {$a$} (1);
\draw[->] (0) edge[loop above,looseness=8] node[above] {$a$} (0); 
\draw[->] (01) edge[loop above,looseness=8] node[above] {$a$} (01); 
\draw[->] (1) edge[loop above,looseness=9] node[above] {$a$} (1); 
\end{tikzpicture}
\qquad \raisebox{10pt}{$\stackrel{p_c}{\Imp}$}\qquad
\begin{tikzpicture}
\node (0) at (0,0) {$\ov{p}_c$};
\node (01) at (1.7,0) {$\dot{p}_c$};
\node (1) at (3.4,0) {$p_c$};
\draw[->] (0) to node[above] {$a$} (01);
\draw[<->] (01) to node[above] {$a$} (1);
\draw[->,bend right =20] (0) to node[below] {$a$} (1);
\draw[->] (01) edge[loop above,looseness=8] node[above] {$a$} (01); 
\draw[->] (1) edge[loop above,looseness=9] node[above] {$a$} (1); 
\end{tikzpicture}
\qquad \raisebox{10pt}{$\stackrel{\neg p_c}{\Imp}$}\qquad
\raisebox{15pt}{
\begin{tikzpicture}
\node (0) at (0,0) {$\ov{p}_c$};
\node (01) at (1.7,0) {$\dot{p}_c$};
\node (1) at (3.4,0) {$p_c$};
\draw[->] (0) to node[above] {$a$} (01);
\draw[<-] (01) to node[above] {$a$} (1);
\draw[->] (01) edge[loop above,looseness=8] node[above] {$a$} (01); 
\end{tikzpicture}
}

\medskip

Now, Anne is not crazy. In principle we can recover from such information states with local inconsistency by having Cath reveal her card to Anne, that will either way make her {\em revise her beliefs}. But we cannot do this as `hard information' in the form of public announcements of $p_c \et \circ p_c$ respectively $\neg p_c \et \circ \neg p_c$, as that would result in an epistemic model with the empty relation for agent $a$. We would simply get this:

\medskip

\begin{tikzpicture}
\node (0) at (0,0) {$\ov{p}_c$};
\node (01) at (1.7,0) {$\dot{p}_c$};
\node (1) at (3.4,0) {$p_c$};
\end{tikzpicture}

\medskip

However, there is an alternative, updating with `soft information'. A well-honoured mechanism for such belief revision in a truth-valued logic is to enrich the structures with preferences or plausibilities among the states \cite{jfaketal.handbook:2015,baltagetal.tlg3:2008,hvd.prolegomena:2005}. Without getting into details, one should imagine the above depicted relations as relations used to interpret belief, whereas there is an additional relation representing knowledge, which in this case would be the universal relation. For a so-called soft revision with $* p_c$ (`*' to denote `revision') we would then get an update as below (returning the epistemic model before the $\neg p_c$ announcement), which can be followed by a subsequent change of mind (of the environment) that is a revision with $* \neg p_c$. Note that such belief revision does not involve factual change. If $p_c$ was initially true, it remains so throughout any update; and if originally false, the same. We do not change the point of evaluation.

\medskip

\noindent
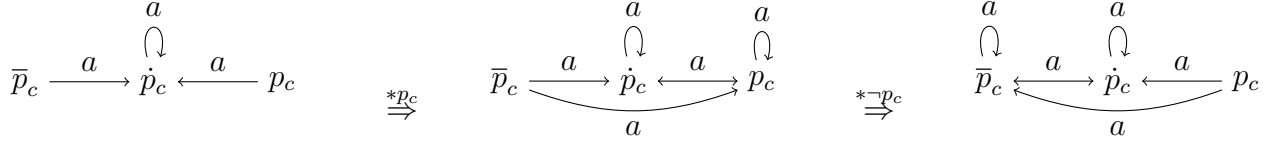
\raisebox{15pt}{
\begin{tikzpicture}
\node (0) at (0,0) {$\ov{p}_c$};
\node (01) at (1.7,0) {$\dot{p}_c$};
\node (1) at (3.4,0) {$p_c$};
\draw[->] (0) to node[above] {$a$} (01);
\draw[<-] (01) to node[above] {$a$} (1);
\draw[->] (01) edge[loop above,looseness=8] node[above] {$a$} (01); 
\end{tikzpicture}
}
\qquad \raisebox{10pt}{$\stackrel{* p_c}{\Imp}$}\qquad
\begin{tikzpicture}
\node (0) at (0,0) {$\ov{p}_c$};
\node (01) at (1.7,0) {$\dot{p}_c$};
\node (1) at (3.4,0) {$p_c$};
\draw[->] (0) to node[above] {$a$} (01);
\draw[<->] (01) to node[above] {$a$} (1);
\draw[->,bend right =20] (0) to node[below] {$a$} (1);
\draw[->] (01) edge[loop above,looseness=8] node[above] {$a$} (01); 
\draw[->] (1) edge[loop above,looseness=9] node[above] {$a$} (1); 
\end{tikzpicture}
\qquad \raisebox{10pt}{$\stackrel{* \neg p_c}{\Imp}$}\qquad
\begin{tikzpicture}
\node (0) at (0,0) {$\ov{p}_c$};
\node (01) at (1.7,0) {$\dot{p}_c$};
\node (1) at (3.4,0) {$p_c$};
\draw[<->] (0) to node[above] {$a$} (01);
\draw[<-] (01) to node[above] {$a$} (1);
\draw[<-,bend right =20] (0) to node[below] {$a$} (1);
\draw[->] (0) edge[loop above,looseness=8] node[above] {$a$} (0); 
\draw[->] (01) edge[loop above,looseness=8] node[above] {$a$} (01); 
\end{tikzpicture}

\medskip

All this becomes even a bit more complex with Cath also modelled in the system. We assume that Cath knows her own state (her own card). As common in dynamic epistemic logic, we model Cath's announcement of any $\alpha$ as the public announcement of $K_c \alpha$.  This would result in the model transitions as below, where only the relations for Anne change and not those for Cath (as she chooses to lie, her relations are not affected). (Note that, for the first announcement, $K_c p_c$ is indeed true in the $p_c$ and in the $\dot{p}_c$ states, as we have two designated values), as in the {\em agent announcements} of \cite{hvd.lying:2014}.

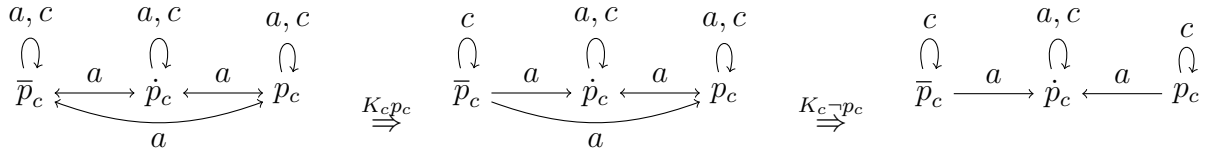
\medskip

\noindent
\begin{tikzpicture}
\node (0) at (0,0) {$\ov{p}_c$};
\node (01) at (1.7,0) {$\dot{p}_c$};
\node (1) at (3.4,0) {$p_c$};
\draw[<->] (0) to node[above] {$a$} (01);
\draw[<->] (01) to node[above] {$a$} (1);
\draw[<->,bend right =20] (0) to node[below] {$a$} (1);
\draw[->] (0) edge[loop above,looseness=8] node[above] {$a,c$} (0); 
\draw[->] (01) edge[loop above,looseness=8] node[above] {$a,c$} (01); 
\draw[->] (1) edge[loop above,looseness=9] node[above] {$a,c$} (1); 
\end{tikzpicture}
\quad \raisebox{10pt}{$\stackrel{K_cp_c}{\Imp}$}\quad
\begin{tikzpicture}
\node (0) at (0,0) {$\ov{p}_c$};
\node (01) at (1.7,0) {$\dot{p}_c$};
\node (1) at (3.4,0) {$p_c$};
\draw[->] (0) to node[above] {$a$} (01);
\draw[<->] (01) to node[above] {$a$} (1);
\draw[->,bend right =20] (0) to node[below] {$a$} (1);
\draw[->] (0) edge[loop above,looseness=8] node[above] {$c$} (0); 
\draw[->] (01) edge[loop above,looseness=8] node[above] {$a,c$} (01); 
\draw[->] (1) edge[loop above,looseness=9] node[above] {$a,c$} (1); 
\end{tikzpicture}
\quad \raisebox{10pt}{$\stackrel{K_c\neg p_c}{\Imp}$}\quad
\raisebox{15pt}{
\begin{tikzpicture}
\node (0) at (0,0) {$\ov{p}_c$};
\node (01) at (1.7,0) {$\dot{p}_c$};
\node (1) at (3.4,0) {$p_c$};
\draw[->] (0) to node[above] {$a$} (01);
\draw[<-] (01) to node[above] {$a$} (1);
\draw[->] (0) edge[loop above,looseness=8] node[above] {$c$} (0); 
\draw[->] (01) edge[loop above,looseness=8] node[above] {$a,c$} (01); 
\draw[->] (1) edge[loop above,looseness=8] node[above] {$c$} (1); 
\end{tikzpicture}
}

\medskip

And again we can recover from inconsistency by dynamic epistemic belief revision, and also consider factual change and relational expansion, similarly to already discussed for the domain restricting updates. Although one can proceed in this way and thus get logics for paraconsistent belief revision instead of paraconsistent knowledge revision, this goes in the direction of more and more complex encodings of structures and accompanying logics, and does not benefit from the advantages of paraconsistency. The beauty of paraconsistency, we think, is that it provides us with alternatives for such belief revision that are more economical in their encodings, and permissive of local inconsistencies, in particular in a multi-agent setting. 

We leave this for future research, as well as the full analysis of the Coup game.

\section{Conclusion}\label{conclusion}
We presented sound and strongly complete axiomatizations for the dynamic epistemic paraconsistent logics AMLFI1 and UMLFI1, showing decidability and compactness in these cases, and have shown that the PALFI1 fragment of AMLFI1 corresponds to another public announcement paraconsistent logic PALFI1$_C$. The technique of reduction axioms via strict equivalences shows a promising direction to tackle dynamic epistemic extensions of many-valued modal logics in general, to be considered in future work.
A detailed analysis of multi-agent system dynamics also involving lying and belief revision concluded our work. We wish to further pursue the modelling of lying, consider belief contraction (adding arrows), and generalize our work to other many-valued logics.

\bibliographystyle{plain}
\bibliography{biblio2026}

\end{document}